\documentclass[10pt,onecolumn]{article}
\usepackage[margin=1in,nohead]{geometry}
\usepackage{algorithm}
\usepackage[noend]{algpseudocode}
\usepackage{amsmath}
\usepackage{amssymb}
\usepackage{amsthm}
\usepackage{bm}
\usepackage{booktabs}
\usepackage{enumitem}
\usepackage{fix-cm}
\usepackage{mathtools}
\usepackage{pgfplots}
\usepackage{siunitx}
\usepackage[labelformat=simple]{subcaption}
\usepackage[hyphens]{url}
\usepackage[hidelinks]{hyperref}
\usepackage[hyphenbreaks]{breakurl}
\usepackage{xfrac}
\newtheorem{theorem}{Theorem}
\newtheorem*{theorem*}{Theorem}
\newtheorem{lemma}{Lemma}

\newtheorem{proposition}{Proposition}
\newtheorem{observation}{Observation}
\newtheorem{definition}{Definition}
\newtheorem{claim}{Claim}
\usepackage{tikz}
\usetikzlibrary{arrows,calc,positioning,trees,shapes,matrix,backgrounds,fit}
\pgfplotsset{compat=newest}
\pgfdeclarelayer{background}
\pgfsetlayers{background,main}
\DeclareMathOperator{\supp}{supp}

\DeclarePairedDelimiter{\norm}{\lVert}{\rVert}
\renewcommand{\vec}[1]{\mathbf{#1}}
\renewcommand{\bar}[1]{\overline{#1}}
\newcommand{\mat}[1]{\ensuremath{#1}}
\newcommand{\binr}[1]{\check{#1}}
\newcommand{\transpose}{\mathsf{T}}
\newcommand{\mdot}{|[fill=black]|} 
\newcommand{\tc}{\ensuremath{TC}}
\newcommand{\bc}{\ensuremath{BC}}
\newcommand{\cc}{\ensuremath{CC}}
\newcommand{\dc}{\ensuremath{DC}}
\newcommand{\ev}{\ensuremath{EV}}
\newcommand{\pr}{\ensuremath{PR}}
\newcommand{\cntvec}[1]{\vec{r}_{#1}}

\title{Triangle Centrality}
\author{Paul Burkhardt \footnote{Research Directorate, National Security Agency,
Fort~Meade, MD 20755. Email: pburkha@nsa.gov}}
\date{October 1, 2024}

\begin{document}
\maketitle

\begin{abstract}
  Triangle centrality is introduced for finding important vertices in a graph
  based on the concentration of triangles surrounding each vertex. It has the
  distinct feature of allowing a vertex to be central if it is in many triangles
  or none at all.

  Given a simple, undirected graph $G=(V,E)$ with $n=|V|$ vertices and $m=|E|$
  edges, let $\triangle(v)$ and $\triangle(G)$ denote the respective triangle
  counts of $v$ and $G$. Let $N(v)$ be the neighborhood set of
  $v$. Respectively, $N_\triangle(v)$ and $N_\triangle[v] = \{v\}\cup
  N_\triangle(v)$ denote the set of neighbors that are in triangles with $v$ and
  the closed set including $v$. Then the triangle centrality for a vertex $v$ is

  \begin{equation*}
    \tc(v) = \frac{\frac{1}{3}\sum_{u\in N_\triangle[v]} \triangle(u)
      + \sum_{w\in \{N(v)\setminus N_\triangle(v)\}} \triangle(w)}
    {\triangle(G)}.
  \end{equation*}

  We show experimentally that triangle centrality is broadly applicable to many
  different types of networks. Our empirical results demonstrate that 30\% of
  the time triangle centrality identified central vertices that differed with
  those found by five well-known centrality measures, which suggests novelty
  without being overly specialized. It is also asymptotically faster to compute
  on sparse graphs than all but the most trivial of these other measures.

  We introduce optimal algorithms that compute triangle centrality in
  $O(m\bar\delta)$ time and $O(m+n)$ space, where $\bar\delta\le O(\sqrt{m})$ is
  the \emph{average degeneracy} introduced by Burkhardt, Faber, and Harris
  (2020). In practical applications, $\bar\delta$ is much smaller than
  $\sqrt{m}$ so triangle centrality can be computed in nearly linear time. On a
  Concurrent Read Exclusive Write (CREW) Parallel Random Access Machine (PRAM),
  we give a near work-optimal parallel algorithm that takes $O(\log n)$ time
  using $O(m\sqrt{m})$ CREW PRAM processors. In MapReduce, we show it takes four
  rounds using $O(m\sqrt{m})$ communication bits and is therefore optimal. We
  also derive a linear algebraic formulation of triangle centrality which can be
  computed in $O(m\bar\delta)$ time on sparse graphs.

  \textit{Keywords}: graph, triangle, centrality, algorithm, parallel, pram,
  mapreduce
\end{abstract}

\section{\label{sec:introduction}Introduction}
Given a network of entities, namely a simple, undirected graph $G=(V,E)$ with
$n=|V|$ vertices connected by $m=|E|$ edges, we wish to know the most important
vertices in the graph. This has many applications in web search, network
routing, or simply identifying influential individuals in society and media. We
introduce a new graph centrality measure that models influence by the
concentration of triangles surrounding a vertex. Here, a triangle is a complete
graph of three vertices and is both a 3-clique and a 3-cycle.

Our new triangle centrality~\cite{bib:nsa2017} has the distinct feature of
identifying importance using the cohesive subnetwork of a vertex without
requiring that vertex to have a high degree of direct connectivity. It is based
on the partitioned sum of triangle counts between the triangle neighbors and
non-triangle neighbors of a vertex, normalized over the total triangle count of
the graph. A vertex is important if it is at the center of many triangles
between itself and its neighbors. The values for triangle centrality are bounded
in the interval $[0,1]$ to indicate the proportion of triangles centered at a
vertex.

Triangle centrality is unique because of its duality principle of allowing a
vertex to be central if it is in many triangles or none at all. If a vertex is
in few to no triangles, it can still rank highly given the contribution from its
non-triangle neighbors. This distinguishes our measure from those based on the
clustering coefficient because a vertex that is not in any triangles has zero
clustering coefficient and therefore cannot be central in such measures.

Our centrality measure captures influence in networks with hierarchical
structure, such as those where organizational leaders delegate actions to a
small number of direct subordinates, who by nature of their roles are more
highly cohesive. For an explicit example, consider a hierarchical organization
such as a corporation in which the Chief Executive Officer (CEO) has a very
small number of direct subordinates. But these subordinates, e.g., other
officers in the company, have mutual associates such as office managers,
administrators, and staff that assist them in carrying out the CEO's
objectives. Thus, the top of the organization may have very few contacts, but
the subordinate contacts are likely connected because they must share similar
information or orders handed down by the top. In turn, these subordinates may
have many mutual associates that are also members of triangles. Thus, the top
vertex in the hierarchy could have few or no triangles, but if there are many
triangles concentrated around its direct contacts and their contacts, it
suggests that top vertex is important in the network.

We also believe triangle centrality is more robust to noise and adversarial
gaming because it requires the cooperation from a pair of connected vertices to
contribute to the rank. A vertex that does not have triangles cannot contribute
to the importance of its neighbors. In contrast, an adversary or cheater could
inflate the rank in measures that depend heavily on direct neighbors by spamming
or creating many spurious links.

We will give a precise, mathematical definition for our triangle
centrality. Clearly, our centrality is of no use in triangle-free graphs, such
as bipartite graphs and trees. We show experimentally that triangle centrality
is broadly applicable to many different types of networks. Our empirical results
demonstrate that 30\% of the time triangle centrality identified central
vertices that differed with those found by five well-known centrality measures,
indicating novelty without being overly specialized. We believe it strikes a
good balance between normality and novelty, and it is asymptotically faster to
compute on sparse graphs than all but the most trivial of these other measures.

The runtime for computing triangle centrality is bounded by counting
triangles. It is known that an optimal triangle counting algorithm takes
$O(m\bar\delta) = O(ma)$ time, where $a$ is the \emph{arboricity} and
$\bar\delta \le 2a \le O(\sqrt{m})$ is the \emph{average degeneracy} introduced
in 2020 by Burkhardt et al.~\cite{bib:burkhardt2020}. We will give optimal
algorithms for computing triangle centrality in $O(m\bar\delta)$ time and
$O(m+n)$ space. In practical applications, $\bar\delta$ is much smaller than
$\sqrt{m}$ so triangle centrality can be computed in nearly linear time. We also
give parallel algorithms for Concurrent Read Exclusive Write (CREW) Parallel
Random Access Machine (PRAM) and MapReduce models. Our CREW algorithm is nearly
work-optimal, taking $O(\log n)$ time and $O(m\sqrt{m})$ processors. Our
MapReduce algorithm takes $O(1)$ rounds and communicates $O(m\sqrt{m})$ bits and
is therefore optimal. We also derive a linear algebraic formulation of triangle
centrality which can be computed in $O(m\bar\delta)$ time on sparse graphs.

We summarize our contribution in Section~\ref{sec:contribution} and set the
motivation in Section~\ref{sec:motivation}. We define triangle centrality and
introduce basic observations in Sections~\ref{sec:formulation} and
\ref{sec:observation}. We compare our triangle centrality to common centrality
measures in Sections~\ref{sec:related} and \ref{sec:compare}. Optimal algorithms
are given in Sections~\ref{sec:algorithm} and \ref{sec:parallel_alg}, followed
by runtime performance in Section~\ref{sec:performance}.

\section{\label{sec:notation}Notation}
The vertices in $G$ are labeled $[n]=1,2,3...n$. A complete graph of $n$
vertices is an $n$-clique and is denoted by $K_n$. The neighborhood (adjacency)
of a vertex $v$ is $N(v)=\{u \mid (u,v) \in E \}$, and the number of neighbors
of a vertex is its degree $d(v)=|N(v)|$. A path is a consecutive sequence of
adjacent edges in $E$. The distance $d_G(u,v)$ is the shortest-path length
between vertices $v$ and $u$.

The average degeneracy $\bar\delta$ of $G$ is defined as

\[
\bar\delta = \frac{1}{m}\sum_{(v,u)\in E} \min\{d(v),d(u)\},
\]
\noindent
then it follows from the arboricity~\cite{bib:chiba_nishizeki1985} that
$\bar\delta \le 2a \le O(\sqrt{m})$. Recall the arboricity is the minimum number
of disjoint forests whose union covers $E$.

Let $\triangle(v)$ denote the triangle count for a vertex $v$ and $\triangle(G)$
for the total triangle count in $G$. The triangle neighborhood of $v$ is given
by $N_\triangle (v) = \{u\in N(v) \mid N(u) \cap N(v) \ne \varnothing\}$, where
the closed triangle neighborhood of $v$ is $N_\triangle[v] = \{v\} \cup
N_\triangle(v)$. We will refer to $N_\triangle[v]$ as the triangle core of $v$.

Let $\pi : V \rightarrow [n]$ be an ordering on vertices in $G$ such $\pi(u) <
\pi(v)$ if $d(u) < d(v)$ or in the case of a tie then $u < v$. We will say that
vertex $u$ is a higher-ordered or higher-ranked vertex than $v$ if $\pi(u) >
\pi(v)$. We denote $N_H(v) = \{u \in N(v) \mid \pi(u) > \pi(v)\}$ as the
higher-ordered neighborhood of $v$ and similarly $N_L(v)$ for the set containing
$\pi(u) < \pi(v)$ lower-ordered neighbors.

The $(ij)$th entry of a matrix $\mat{X}$ may be written as $x(i,j)$, $x_{ij}$,
or also $X_{ij}$. The $j$th column vector is written $\vec{X_j}$, and the
corresponding row vector is given by the transpose $\vec{X_j}^\transpose$. The
$\ell_1$ norm of a vector $\vec{x}$ is denoted $\norm{\vec{x}}_1=\sum_i \lvert
x_i \rvert$. The support of a vector $\vec{x}$, denoted by $\supp \vec{x}$,
refers to the set of indices corresponding to non-zero entries in $\vec{x}$,
thus $\supp \vec{x} = \{i \mid x(i) \ne 0\}$. We use $\circ$ to denote the
Hadamard product operator for elementwise matrix multiplication. The symbol
$\omega$ denotes the matrix multiplication exponent. The binary (Boolean) form
of a matrix (or vector), where non-zeros take the value of 1, is written as
$\binr{\mat{X}}$ in the case of a matrix $\mat{X}$ and similarly for a vector.

Let $\mat{A} \in \{0,1\}^{n \times n}$ be the symmetric adjacency matrix for
$G$. The graph triangle matrix $\mat{T}=\mat{A}^2\circ \mat{A}$ holds the
triangle counts between a vertex and its neighbors~\cite{bib:burkhardt2017}. The
vector $\vec{1}$ is a vector of all ones, and $\mat{I}$ is the identity matrix.

We will use the abbreviations in Table~\ref{tbl:abbrev} for the different
centrality measures appearing in this article. For a centrality measure $i$, let
$\cntvec{i}$ denote a vector holding the centrality score of every vertex in $G$
(e.g., $\cntvec{\ev}$).

\begin{table}[H]
  \caption{\label{tbl:abbrev}Centrality Glossary}
  \centering
  \begin{tabular}{ll}
    \toprule
    \tc & Triangle centrality \\ 
    \bc & Betweenness centrality \\
    \cc & Closeness centrality \\
    \dc & Degree centrality \\
    \ev & Eigenvector centrality \\
    \pr & PageRank centrality \\
    \bottomrule
  \end{tabular}
\end{table}

Table~\ref{tbl:symbol} summarizes the mathematical symbols used throughout this
article.

\begin{table}[H]
  \caption{\label{tbl:symbol}Symbol Glossary}
  \centering
  \begin{tabular}{ll}
    \toprule
    $G$ & Simple undirected graph \\
    $K_n$ & Complete graph on $n$ vertices (n-clique) \\
    $n$ & Vertex count of $G$ \\
    $m$ & Edge count of $G$ \\
    $d(v)$ & Degree of vertex $v$ \\
    $d_G(u,v)$ & Distance between vertices $u,v$ in $G$ \\
    $a$ & Arboricity of $G$ \\
    $\bar\delta$ & Average degeneracy of $G$ \\
    $\triangle(G)$ & Total triangle count in $G$ \\
    $\triangle(v)$ & Local triangle count of vertex $v$ \\
    $N(v)$ & Neighborhood of vertex $v$ \\
    $N_\triangle (v)$ & Triangle neighborhood of vertex $v$ \\
    $N_\triangle[v]$ & Closed triangle neighborhood of vertex $v$
    (triangle core) \\
    $N_H(v)$ & Higher-ordered neighborhood of vertex $v$ \\
    $N_L(v)$ & Lower-ordered neighborhood of vertex $v$ \\
    $\mat{X}$ & Matrix notation \\
    $\vec{x}$ & Vector notation \\
    $x(i,j)$ (or $x_{ij},X_{ij}$) & $(ij)$th entry of matrix $\mat{X}$ \\
    $\vec{X_j}$ & $j$th column vector of matrix $\mat{X}$ \\
    $\transpose$ & Transpose, e.g., $\mat{X}^\transpose$ (matrix transpose) \\
    $\supp \vec{x}$ & Support of vector $\vec{x}$ \\
    $\norm{\vec{x}}_1$ & $\ell_1$ norm of vector $\vec{x}$ \\
    $\circ$ & Hadamard product operator (element-wise multiplication) \\
    $\omega$ & Matrix multiplication exponent \\
    $\vec{1}$ & Vector of all ones \\
    $\mat{I}$ & Identity matrix \\
    $\mat{A}$ & Adjacency matrix of $G$ \\
    $\mat{T}$ & Graph triangle matrix ($\mat{T}=\mat{A}^2\circ \mat{A}$) \\
    $\binr{\mat{T}}$ & Binary graph triangle matrix \\
    \bottomrule
  \end{tabular}
\end{table}

\section{\label{sec:contribution}Contribution}
We introduce triangle centrality for identifying important vertices in a
graph. It has the following advantages:

\begin{itemize}
\item duality principle allowing high- or low-triangle counts;

\item influence does not depend on a high number of direct contacts;
  
\item finishes in a constant number of steps (non-iterative);

\item fast asymptotic runtime equivalent to that of triangle counting.
\end{itemize}

Triangle centrality is formally defined by this next definition followed by a
linear algebraic formulation.

\newtheorem*{def:tricent}{Definition~\ref{def:tricent}}
\begin{def:tricent}
  The triangle centrality of a vertex $v$ is

  \begin{equation*}
    \tc(v) = \frac{\frac{1}{3}\sum_{u\in N_\triangle[v]} \triangle(u) +
      \sum_{w\in \{N(v)\setminus N_\triangle(v)\}} \triangle(w)}
    {\triangle(G)}.
  \end{equation*}
\end{def:tricent}

\newtheorem*{prop:tricent_algebraic}{Proposition~\ref{prop:tricent_algebraic}}
\begin{prop:tricent_algebraic}
  The linear algebraic triangle centrality for all vertices is given by,

  \begin{equation*}
    \cntvec{\tc} = \frac{\left(3\mat{A} - 2\binr{\mat{T}} + \mat{I}\right)
      \mat{T}\vec{1}}{\vec{1}^\transpose \mat{T} \vec{1}}.
  \end{equation*}
\end{prop:tricent_algebraic}

Since triangle centrality depends only upon triangle counts, we show that
triangle centrality can be optimally computed leading to the following
algorithmic results.

\newtheorem*{thm:tricent}{Theorem~\ref{thm:tricent}}
\begin{thm:tricent}
  Triangle centrality can be computed in $O(m\bar\delta)$ time and $O(m+n)$
  space for all vertices in $G$.
\end{thm:tricent}

\newtheorem*{thm:tricent_algebraic}{Theorem~\ref{thm:tricent_algebraic}}
\begin{thm:tricent_algebraic}
  There is a linear algebraic algorithm that computes the triangle centrality in
  $O(m\bar\delta)$ time and $O(m+n)$ space for all vertices in $G$ given a
  sparse matrix $\mat{A}$.
\end{thm:tricent_algebraic}

\newtheorem*{thm:tricent_fastmatrix}{Theorem~\ref{thm:tricent_fastmatrix}}
\begin{thm:tricent_fastmatrix}
  Triangle centrality can be computed in $n^{\omega+o(1)}$ time using fast
  matrix multiplication, where $\omega$ is the matrix multiplication exponent,
  for all vertices in $G$.
\end{thm:tricent_fastmatrix}

We also give parallel algorithms to compute the triangle centrality on a CREW
PRAM and in MapReduce. The MapReduce algorithm is optimal, and the CREW PRAM is
optimal up to a logarithmic factor, yielding the next results.

\newtheorem*{thm:tricent_crew}{Theorem~\ref{thm:tricent_crew}}
\begin{thm:tricent_crew}
  Triangle centrality can be computed on a CREW PRAM in $O(\log n)$ time using
  $O(m\sqrt{m})$ processors for all vertices in $G$.
\end{thm:tricent_crew}

\newtheorem*{thm:tricent_mapreduce}{Theorem~\ref{thm:tricent_mapreduce}}
\begin{thm:tricent_mapreduce}
  Triangle centrality can be computed using MapReduce in four rounds and
  $O(m\sqrt{m})$ communication bits for all vertices in $G$.
\end{thm:tricent_mapreduce}

\section{\label{sec:motivation}Motivation}
We use triangle counts for our centrality measure because the role of triangles
in network cohesion has been well established in social network
analysis~\cite{bib:watts_strogatz1998, bib:newman2001, bib:newman2002,
  bib:newman_park2003, bib:friggeri2011}. If an individual has two friends and
these two friends are also friends, then the trio is more cohesive. A
concentration of triangles increases network density, allowing information to
spread more rapidly because there are more connected pathways. The importance of
triangles in cohesive networks was formally developed by Friggeri et al. in
2011~\cite{bib:friggeri2011}, but cohesive networks based on triangles were
explored earlier with the introduction of the $k$-truss by Cohen in
2005~\cite{bib:cohen2005}, later published in 2009~\cite{bib:cohen2009}. The
$k$-truss is a maximal subgraph in which each edge is incident to $k$
triangles. See~\cite{bib:burkhardt2020} for a more complete analysis of graph
trusses. In 1998, Watts and Strogatz~\cite{bib:watts_strogatz1998} found that
triangles were integral to the property of real-world networks and introduced
the clustering coefficient as a measure of how likely a pair of neighbors of a
vertex may themselves be directly connected. Triangles are also a key component
of clusters in social networks~\cite{bib:newman2001, bib:newman2002,
  bib:newman_park2003}.

We claim that importance is due to the concentration of triangles around a
vertex, allowing that vertex to either be involved in many triangles or
not. Meaning, an important vertex is at the center of many triangles, but it may
or may not be incident to many triangles. Our goal is to capture this
counter-intuitive duality, which appears in networks exhibiting organizational
hierarchies where both direct and indirect connectivity need to be
accounted. Triangle centrality is able to capture important figures who are
either very active in the network or inactive remaining ``behind-the-scenes''
delegating to subordinates. To the best of our knowledge, this duality is not
explicitly modeled by other centrality measures.

For a vertex $v$ to be important in our measure, it requires the support of a
pair of adjacent vertices $\{u,w\}$ that must either be in a triangle with $v$
or in a triangle with a neighbor of $v$. We posit if $v$ has neighbors $u$ that
are involved in many triangles and hence are themselves important, then these
$u$ affirm $v$ is also important regardless of whether $v$ itself is in any
triangles. In a real-world scenario, this could mean that the original pair of
connected vertices have conferred and agreed upon the selection of the mutual
vertex as a contact. This differs starkly from other centrality measures where
the contribution or ``vote'' for a vertex comes from neighbors regardless of the
connection between neighbors. Thus, an e-mail spammer can appear as an important
vertex in these measures.

To illustrate our premise, we argue that vertex ``a'' in each of the graphs in
Figure~\ref{fig:tc_ex} is the most triangle-centric vertex and henceforth the
most important vertex. We will compare the ranking of vertex ``a'' in our
discussion of other centrality measures at the end of Section~\ref{sec:related},
and later in Section~\ref{sec:compare}, we will analyze triangle centrality
results with these other measures on 20 realistic graphs.

\begin{figure}[H]
\centering
\begin{subfigure}[t]{0.35\textwidth}
  \centering
  \begin{tikzpicture}
    [
      scale=0.75, inner sep=2.5,
      every node/.style = {circle,draw=black,thin,fill=white}
    ]
    \node (a) {a};
    \node (b) at ([shift={(120:1.5)}]a) {};
    \node (c) at ([shift={(60:1.5)}]a) {};
    \node (d) at ([shift={(0:1.5)}]a) {};
    \node (e) at ([shift={(300:1.5)}]a) {};
    \node (f) at ([shift={(240:1.5)}]a) {};
    \node (g) at ([shift={(180:1.5)}]a) {};

    \node (1) at ([shift={(75:0.75)}]b) {};
    \node (2) at ([shift={(105:0.75)}]b) {};
    \node (3) at ([shift={(135:0.75)}]b) {};
    \node (4) at ([shift={(165:0.75)}]b) {};

    \node (5) at ([shift={(105:0.75)}]c) {};
    \node (6) at ([shift={(75:0.75)}]c) {};
    \node (7) at ([shift={(45:0.75)}]c) {};
    \node (8) at ([shift={(15:0.75)}]c) {};

    \node (9) at ([shift={(45:0.75)}]d) {};
    \node (10) at ([shift={(15:0.75)}]d) {};
    \node (11) at ([shift={(345:0.75)}]d) {};
    \node (12) at ([shift={(315:0.75)}]d) {};

    \node (13) at ([shift={(345:0.75)}]e) {};
    \node (14) at ([shift={(315:0.75)}]e) {};
    \node (15) at ([shift={(285:0.75)}]e) {};
    \node (16) at ([shift={(255:0.75)}]e) {};

    \node (17) at ([shift={(285:0.75)}]f) {};
    \node (18) at ([shift={(255:0.75)}]f) {};
    \node (19) at ([shift={(225:0.75)}]f) {};
    \node (20) at ([shift={(195:0.75)}]f) {};

    \node (21) at ([shift={(225:0.75)}]g) {};
    \node (22) at ([shift={(195:0.75)}]g) {};
    \node (23) at ([shift={(165:0.75)}]g) {};
    \node (24) at ([shift={(135:0.75)}]g) {};

    \draw (a) to (b);
    \draw (a) to (c);
    \draw (a) to (d);
    \draw (a) to (e);
    \draw (a) to (f);
    \draw (a) to (g);
    \draw (b) to (c);
    \draw (d) to (e);
    \draw (f) to (g);

    \draw (b) to (1);
    \draw (b) to (2);
    \draw (b) to (3);
    \draw (b) to (4);
    \draw (c) to (5);
    \draw (c) to (6);
    \draw (c) to (7);
    \draw (c) to (8);
    \draw (d) to (9);
    \draw (d) to (10);
    \draw (d) to (11);
    \draw (d) to (12);
    \draw (e) to (13);
    \draw (e) to (14);
    \draw (e) to (15);
    \draw (e) to (16);
    \draw (f) to (17);
    \draw (f) to (18);
    \draw (f) to (19);
    \draw (f) to (20);
    \draw (g) to (21);
    \draw (g) to (22);
    \draw (g) to (23);
    \draw (g) to (24);
  \end{tikzpicture}
  \subcaption{\label{fig:tc_ex_1}}
\end{subfigure}
~
\begin{subfigure}[t]{0.35\textwidth}
  \centering
  \begin{tikzpicture}
    [
      scale=0.75, inner sep=2.5,
      every node/.style = {circle,draw=black,thin,fill=white}
    ]
    \def \n{6}
    \node (a) {a};

    \foreach \x in {0,...,5} {
      \node (\x)
      at ($ ([shift={(45:2cm)}]a) + (360/\n*\x:0.75) $) {};
    }
    \foreach \x [evaluate=\x as \ystart using int(\x+1)] in {0,...,4} {
      \foreach \y in {\ystart,...,5} {
        \draw (\x) -- (\y);
      }
    }
    \draw (a) -- (4);
    
    \foreach \x in {0,...,5} {
      \node (\x)
      at ($ ([shift={(-45:2cm)}]a) + (360/\n*\x:0.75) $) {};
    }
    \foreach \x [evaluate=\x as \ystart using int(\x+1)] in {0,...,4} {
      \foreach \y in {\ystart,...,5} {
        \draw (\x) -- (\y);
      }
    }
    \draw (a) -- (2);

    \foreach \x in {0,...,5} {
      \node (\x)
      at ($ ([shift={(135:2cm)}]a) + (360/\n*\x:0.75) $) {};
    }
    \foreach \x [evaluate=\x as \ystart using int(\x+1)] in {0,...,4} {
      \foreach \y in {\ystart,...,5} {
        \draw (\x) -- (\y);
      }
    }
    \draw (a) -- (5);

    \foreach \x in {0,...,5} {
      \node (\x)
      at ($ ([shift={(225:2cm)}]a) + (360/\n*\x:0.75) $) {};
    }
    \foreach \x [evaluate=\x as \ystart using int(\x+1)] in {0,...,4} {
      \foreach \y in {\ystart,...,5} {
        \draw (\x) -- (\y);
      }
    }
    \draw (a) -- (1);
  \end{tikzpicture}
  \subcaption{\label{fig:tc_ex_2}}
\end{subfigure}

\bigskip
\begin{subfigure}[t]{0.35\textwidth}
  \centering
  \begin{tikzpicture}
    [
      scale=0.75, inner sep=2.5,
      every node/.style = {circle,draw=black,thin,fill=white}
    ]
    \node (1) {};
    \node (2) at ([shift={(75:1cm)}]1) {};
    \node (3) at ([shift={(105:1cm)}]1) {};
    \node (4) at ([shift={(135:1cm)}]1) {};
    \node (5) at ([shift={(165:1cm)}]1) {};
    \node (6) at ([shift={(195:1cm)}]1) {};
    \node (7) at ([shift={(225:1cm)}]1) {};
    \node (8) at ([shift={(255:1cm)}]1) {};
    \node (9) at ([shift={(285:1cm)}]1) {};
    \draw (1) -- (2);
    \draw (1) -- (3);
    \draw (1) -- (4);
    \draw (1) -- (5);
    \draw (1) -- (6);
    \draw (1) -- (7);
    \draw (1) -- (8);
    \draw (1) -- (9);
    \node (a) [right=1.5cm of 1] {a};
    \node (10) [above=0.5cm of a] {};
    \node (11) at ([shift={(120:1cm)}]10) {};
    \node (12) at ([shift={(60:1cm)}]10) {};
    \node (13) at ([shift={(30:1.15cm)}]a) {};
    \node (14) at ([shift={(330:1.15cm)}]a) {};
    \node (15) [below=0.5cm of a] {};
    \node (16) at ([shift={(240:1cm)}]15) {};
    \node (17) at ([shift={(300:1cm)}]15) {};
    \draw (1) -- (a);
    \draw (a) -- (10);
    \draw (a) -- (13);
    \draw (a) -- (14);
    \draw (a) -- (15);
    \draw (10) -- (11);
    \draw (10) -- (12);
    \draw (11) -- (12);
    \draw (13) -- (14);
    \draw (15) -- (16);
    \draw (15) -- (17);
    \draw (16) -- (17);
  \end{tikzpicture}
  \subcaption{\label{fig:tc_ex_3}}
\end{subfigure}
~
\begin{subfigure}[t]{0.35\textwidth}
  \centering
  \begin{tikzpicture}
    [
      scale=0.75, inner sep=2.5,
      every node/.style = {circle,draw=black,thin,fill=white}
    ]
    \def \n{5}
    \node (a) {a};

    \node (1) at ([shift={(180:1cm)}]a) {};
    \node (2) at ([shift={(108:1cm)}]1) {};
    \node (3) at ([shift={(36:1cm)}]2) {};
    \node (4) at ([shift={(324:1cm)}]3) {};

    \node (5) [below=of a] {};
    \node (6) at ([shift={(240:1cm)}]5) {}; 
    \node (7) at ([shift={(300:1cm)}]5) {};

    \node (8) [right=1.5cm of a]  {};
    \node (9) at ([shift={(105:1cm)}]8) {};
    \node (10) at ([shift={(75:1cm)}]8) {};
    \node (11) at ([shift={(45:1cm)}]8) {};
    \node (12) at ([shift={(15:1cm)}]8) {};
    \node (13) at ([shift={(345:1cm)}]8) {};
    \node (14) at ([shift={(315:1cm)}]8) {};
    \node (15) at ([shift={(285:1cm)}]8) {};
    \node (16) at ([shift={(255:1cm)}]8) {};
    \node (17) at ([shift={(225:1cm)}]8) {};

    \draw (a) -- (1);
    \draw (a) -- (2);
    \draw (a) -- (3);
    \draw (a) -- (4);
    \draw (1) -- (2);
    \draw (1) -- (3);
    \draw (1) -- (4);
    \draw (2) -- (3);
    \draw (2) -- (4);
    \draw (3) -- (4);

    \draw (a) -- (5);
    \draw (5) -- (6);
    \draw (5) -- (7);
    \draw (6) -- (7);

    \draw (a) -- (8);
    \draw (8) -- (9);
    \draw (8) -- (10);
    \draw (8) -- (11);
    \draw (8) -- (12);
    \draw (8) -- (13);
    \draw (8) -- (14);
    \draw (8) -- (15);
    \draw (8) -- (16);
    \draw (8) -- (17);
  \end{tikzpicture}
  \subcaption{\label{fig:tc_ex_4}}
\end{subfigure}
\caption{\label{fig:tc_ex}Vertex ``a'' in each graph is the most central
  vertex.}
\end{figure}

\section{\label{sec:formulation}Formulation}
Our notion of centrality is based on the triangle counts from a vertex and each
of its neighbors. The contribution of triangle counts for the centrality of a
vertex comes from its subgraph of diameter four. We claim a vertex is important
if it is in many triangles or if its neighbors are in many triangles. This
leads us to partition the triangle counts between neighbors of a vertex that are
in triangles with it and those that are not. This allows a vertex to be highly
ranked if it is at the center of many triangles in its subnetwork without having
to be directly incident to triangles. Recall $N_\triangle[v]$ is the triangle
core of $v$, meaning the set containing $v$ and its triangle neighbors. Our
precise definition of triangle centrality is given next.

\begin{definition}
  \label{def:tricent}
  The triangle centrality of a vertex $v$ is

  \begin{equation*}
    \tc(v) = \frac{\frac{1}{3}\sum_{u\in N_\triangle[v]} \triangle(u)
      + \sum_{w\in \{N(v)\setminus N_\triangle(v)\}} \triangle(w)}
    {\triangle(G)}.
  \end{equation*}
\end{definition}

This definition captures the observations that the total triangle count for a
vertex and its neighbors cannot exceed $\triangle(G)$, and conversely, a vertex
and its neighbors may not be in any triangles. The sum of triangle counts over
the triangle core $N_\triangle[v]$ in Definition~\ref{def:tricent} is multiplied
by a $\frac{1}{3}$ factor to prevent overcounting triangles because the same
triangle is counted by each of its three vertices. If the graph is a clique,
then it suggests each vertex is equally central. Definition~\ref{def:tricent}
also leads to convenient expressions for other special cases that we describe in
Section~\ref{sec:observation}. The range of values triangle centrality can take
implies the proportion of triangles centered at a vertex.

\begin{proposition}
  \label{prop:tricent_range}
  The triangle centrality of a vertex takes values in the range $[0,1]$.
\end{proposition}

\begin{proof}
  We begin by analyzing the numerator in Definition~\ref{def:tricent}.
  
  If a vertex and its neighbors are incident to all triangles in $G$, then the
  second term in the numerator is 0 and the first term gives $\triangle(G)$
  because each triangle is counted once.

  If a vertex is not incident to any triangles, then the first term in the
  numerator is 0 and the second term can range from 0 to $\triangle(G)$.

  Consequently, dividing the numerator by $\triangle(G)$ ensures the values are
  bounded in the range $[0,1$].
\end{proof}

An example calculation of Definition~\ref{def:tricent} is given in
Figure~\ref{fig:simple_ex}. Overall there are three triangles in
Figure~\ref{fig:simple_ex}, two from $N_\triangle[v]$ and one from a neighbor of
$v$ that is not in $N_\triangle[v]$. Without the $\frac{1}{3}$ factor the
triangles in $N_\triangle[v]$ would be overcounted, once from each vertex of a
triangle. We argue that $v$ is at the center of the concentration of triangles
in Figure~\ref{fig:simple_ex}, and therefore, its triangle centrality value is
exactly one.

\begin{figure}[H]
  \centering
  \framebox[1.1\width]{%
  \begin{tikzpicture}
    [inner sep=3, node distance=0.5cm,
    vertex/.style={circle,draw=black,thin},
    edge/.style={thick}]
    \node [vertex] (v) {v};
    \node [vertex] (c) [left=of v] {};
    \node [vertex] (a) [above=of c] {};
    \node [vertex] (b) [below=of c] {};
    \node [vertex] (d) [right=of v] {};
    \node [vertex] (e) [above right=of d] {};
    \node [vertex] (f) [below right=of d] {};
    \draw [edge] (v) to (a);
    \draw [edge] (v) to (b);
    \draw [edge] (v) to (c);
    \draw [edge] (a) to (c);
    \draw [edge] (b) to (c);
    \draw [edge] (v) to (d);
    \draw [edge] (d) to (e);
    \draw [edge] (d) to (f);
    \draw [edge] (e) to (f);
    \matrix(m) [xshift=.5cm, right=of d] 
    [matrix of nodes,ampersand replacement=\&,
    column 1/.style={anchor=west}] {
      $\sum_{u\in N_\triangle[v]} \triangle(u)=6$ \\
      $\sum_{w\in \{N(v)\setminus N_\triangle(v)\}} \triangle(w)=1$\\
      $\triangle(G)=3$ \\
      $\tc(v)=\dfrac{\frac{1}{3}(6)+1}{3}=1$\\
    };
  \end{tikzpicture}
}
  \caption{\label{fig:simple_ex}Example calculation of triangle centrality.}
\end{figure}

\subsection{\label{sec:algebraic}Algebraic Form}
The triangle centrality given in Definition~\ref{def:tricent} can also be
formulated in linear algebra using the adjacency matrix $\mat{A}$ and the graph
triangle matrix. This author introduced the graph triangle matrix
in~\cite{bib:burkhardt2017} and defined it as $\mat{T}=\mat{A}^2\circ
\mat{A}$. Thus, $\mat{T}$ encodes the triangle neighbors of each vertex, just as
$\mat{A}$ encodes all neighbors. It follows that the triangle counts can be
obtained from $\mat{T}$.

\begin{theorem*}[Thm.~1~\cite{bib:burkhardt2017}]
  Given $G$ and the Hadamard product $(\mat{A}^2\circ \mat{A})$, then
  $\triangle(v) = \frac{1}{2} \sum_v (\mat{A}^2\circ \mat{A})_v$ and
  $\triangle(G) = \frac{1}{6} \sum_{ij} (\mat{A}^2\circ \mat{A})_{ij}$.
\end{theorem*}

Since $\mat{A}$ is symmetric, then $\mat{T}=\mat{T}^\transpose$. The matrix
$\mat{T}=\mat{A}^2\circ \mat{A}$ cannot have more non-zeros than $\mat{A}$ due
to the Hadamard operation, and the value of a non-zero $t(i,j)$ indicates the
count of triangles incident to the $\{i,j\}$ edge. Hence,

\begin{align*}
  \triangle(v) &= \frac{1}{2} \norm{\vec{T_v}}_1
  = \vec{T_v}^\transpose \vec{1} \\
  \triangle(G) &= \frac{1}{6} \vec{1}^\transpose \mat{T} \vec{1}.
\end{align*}

This leads us to the linear algebraic formulation for triangle centrality.

\begin{proposition}
  \label{prop:tricent_algebraic}
  The linear algebraic triangle centrality for all vertices is given by,

  \begin{equation*}
    \cntvec{\tc} = \frac{\left(3\mat{A} - 2\binr{\mat{T}}
      + \mat{I}\right)\mat{T} \vec{1}}{\vec{1}^\transpose \mat{T} \vec{1}}.
  \end{equation*}
\end{proposition}

\begin{proof}
  The non-zeros in each $\vec{T_v}$ column vector correspond to the triangle
  neighbors of $v$, hence $N_\triangle(v) = \supp \vec{T_v}$. Using this and
  substituting $\triangle(v) = \frac{1}{2} \norm{\vec{T_v}}_1$ and $\triangle(G)
  = \frac{1}{6} \vec{1}^\transpose \mat{T} \vec{1}$ into
  Definition~\ref{def:tricent} yields

  \begin{displaymath}
    \tc(v) = \frac{\sum_{u\in \{v\} \cup\, \supp \vec{T_v}} \norm{\vec{T_u}}_1
      + 3\sum_{w\in \supp (\vec{A_v} - \vec{T_v})} \norm{\vec{T_w}}_1}
    {\vec{1}^\transpose \mat{T} \vec{1}}.
  \end{displaymath}

  Now recall that a matrix-vector multiplication produces a linear combination
  of the matrix column vectors, where each corresponds to an index of a non-zero
  value in the input vector and is scaled by that value. Then summing the
  triangle counts for all triangle neighbors of $v$, hence $\sum_{u\in
    N_\triangle(v)} \norm{\vec{T_u}}_1$, can be achieved by the
  $(\mat{T}\vec{\binr{\mat{T}}_v})^\transpose \vec{1}$ product. The core
  triangle sum $\sum_{u\in \{v\} \cup\, \supp \vec{T_v} } \norm{\vec{T_u}}_1$,
  which includes the triangle count for $v$, can be obtained by
  $(\mat{T}\vec{\binr{\mat{T}}_v} + \vec{T_v})^\transpose \vec{1}$. Similarly,
  to get the sum of triangle counts from non-triangle neighbors, we use
  $\left(\mat{T}(\vec{A_v}-\vec{\binr{\mat{T}}_v})\right)^\transpose
  \vec{1}$. The next steps demonstrate how to transform the triangle centrality
  for a single vertex into matrix-vector products.

  \begin{align*}
    \tc(v) &= \frac{\sum_{u\in \{v\} \cup\, \supp \vec{T_v}} \norm{\vec{T_u}}_1
      + 3\sum_{w\in \supp (\vec{A_v} - \vec{T_v})} \norm{\vec{T_w}}_1}
    {\vec{1}^\transpose \mat{T} \vec{1}} \\
    &= \frac{(\mat{T} \vec{\binr{\mat{T}}_v} + \vec{T_v})^\transpose \vec{1}
      + 3\left(
      \mat{T} (\vec{A_v} - \vec{\binr{\mat{T}}_v})
      \right)^\transpose\vec{1}}
    {\vec{1}^\transpose \mat{T} \vec{1}} \\
    &= \frac{\left(3\vec{A_v} - 2\vec{\binr{\mat{T}}_v}
      + \vec{I_v} \right)^\transpose \mat{T} \vec{1}}
    {\vec{1}^\transpose \mat{T} \vec{1}}.
    \tag{$T=T^\transpose$ and $\vec{T_v}=\mat{T}\vec{I_v}$}
  \end{align*}

  Then extending for all vertices leads to the desired matrix formulation.
\end{proof}

Given $\mat{T}$, the linear algebraic triangle centrality can be computed with
just two matrix-vector multiplications, two matrix additions, and a single inner
product.

\section{\label{sec:observation}Observations}
In some cases, Definition~\ref{def:tricent} can lead to the surprising result
that the triangle centrality for a vertex is not explicitly dependent on the
number of triangles. This can be useful to quickly test and validate
implementations of triangle centrality.

\begin{observation}
  \label{obs:in_clique}
  The triangle centrality for every vertex is $1$ if $G$ is a clique.
\end{observation}

\begin{proof}
  Since $G$ is a clique, then every neighbor of a vertex $v$ is a triangle
  neighbor; also $d(v)=n-1$ for each vertex and $\triangle(G)$ equates to
  $\binom{n}{3}$. By Definition~\ref{def:tricent} this leads to,

  \begin{align*}
    \tc(v) &= \frac{\frac{1}{3}\sum_{u\in N_\triangle[v]} \triangle(u)}
    {\triangle(G)}
    = \frac{\frac{1}{3}\sum_{v\in V} \binom{d(v)}{2}}{\binom{n}{3}} \\
    &= \frac{\frac{n}{3}\binom{n-1}{2}}{\binom{n}{3}}
    = \frac{\binom{n}{3}}{\binom{n}{3}} = 1.\qedhere
  \end{align*}
\end{proof}

\begin{observation}
  \label{obs:bridge_clique}
  The triangle centrality for a vertex is $\frac{3}{k}$ if it is not in
  triangles and its neighbors are in copies of $K_k$ containing all triangles in
  $G$, where $k\ge 3$.
\end{observation}

\begin{proof}
  Since all triangles are in copies of $K_k$ adjacent to $v$, then
  $\triangle(G)=d(v)\binom{k}{3}$, and the sum of triangle counts for the
  neighbors of $v$ is $d(v)\binom{k-1}{2}$. Therefore, the triangle centrality
  for a vertex $v$ whose neighbors are each in a $K_k$ and $v$ itself has no
  triangles is $\frac{3}{k}$ as follows from Definition~\ref{def:tricent}.

  \begin{align*}
    \tc(v) &= \frac{\sum_{u\in N(v)} \triangle(u)}{\triangle(G)}
    = \frac{\binom{k-1}{2}}{\binom{k}{3}} = \frac{3}{k}.\qedhere
  \end{align*}
\end{proof}

\begin{observation}
  \label{obs:disjoint_clique}
  The triangle centrality for a vertex is $\frac{1}{p}$ if it is in one of $p\ge
  1$ disjoint copies of $K_k$ containing all triangles in $G$, where $k\ge 3$.
\end{observation}

\begin{proof}
  Since all triangles are in the $p$ disjoint copies of $K_k$, then
  $\triangle(G)=p\binom{k}{3}$, and any vertex $v$ in a $K_k$ has
  $\triangle(v)=\binom{k-1}{2}$ triangles. By Definition~\ref{def:tricent}, the
  triangle centrality is then,

  \begin{align*}
    \tc(v) &= \frac{\frac{1}{3}\sum_{u\in N_\triangle[v]} \triangle(u)}
    {\triangle(G)}
    = \frac{\frac{k}{3}\binom{k-1}{2}}{\triangle(G)} \\
    &= \frac{\binom{k}{3}}{\triangle(G)}
    = \frac{\binom{k}{3}}{p\binom{k}{3}} = \frac{1}{p}.\qedhere
  \end{align*}
\end{proof}

\begin{observation}
  \label{obs:clique_chain}
  The triangle centrality for a vertex can be one of the following if $G$ is a
  chain of $p\ge 3$ copies of $K_k$, each connected by a single vertex, where
  $k\ge 3$:

  \begin{enumerate}[label={(\arabic*)}]
  \item $\frac{(2k+2)}{pk}$ if joining two internal $K_k$'s
  \item $\frac{(2k+1)}{pk}$ if joining the head or tail $K_k$
  \item $\frac{(k+2)}{pk}$ does not join $K_k$'s and is not in head or tail
    $K_k$
  \item $\frac{(k+1)}{pk}$ does not join $K_k$'s and is in head or tail $K_k$
  \end{enumerate}
\end{observation}

\begin{proof}
  We begin with some preliminary facts. Any vertex $v$ in a $K_k$ that does not
  join another $K_k$ has degree $d(v)=k-1$ and $\triangle(v)=\binom{k-1}{2}$
  triangles. Any vertex $v$ that joins two $K_k$'s has degree $d(v)=2(k-1)$ and
  $\triangle(v)=2\binom{k-1}{2}$. Since all triangles are in the $K_k$'s, then
  $\triangle(G) = p\binom{k}{3}$. We will use the next relation to prove each of
  the cases.

  \begin{displaymath}
    \frac{1}{3\triangle(G)}\binom{k-1}{2} =
    \frac{\binom{k-1}{2}}{3p\binom{k}{3}} = \frac{1}{pk}.
  \end{displaymath}

  In case (1), a vertex $v$ joins two internal $K_k$'s and since $G$ is a chain
  of at least four $K_k$'s, then two neighbors also join $K_k$ and these will
  have the same number of triangles as $v$. There are $d(v)-2=2(k-2)$ remaining
  neighbors. The triangle centrality for $v$ is,

  \begin{align*}
    \tc(v) &= \frac{1}{3\triangle(G)}\binom{k-1}{2}\Bigl(2+4+2(k-2)\Bigr) \\
    &= \frac{(2k+2)}{pk}.
  \end{align*}

  In case (2), a vertex $v$ joins a head or tail $K_k$ to the chain. It has one
  neighbor that also joins $K_k$'s and $d(v)-1=2k-3$ remaining neighbors. The
  triangle centrality for $v$ is,

  \begin{align*}
    \tc(v) &= \frac{1}{3\triangle(G)}\binom{k-1}{2}\Bigl(2+2+2k-3\Bigr) \\
    &= \frac{(2k+1)}{pk}.
  \end{align*}

  In case (3), a vertex $v$ does not join any $K_k$ and is not in the head or
  tail $K_k$. Two of its neighbors join $K_k$'s so there are $d(v)-2=k-3$
  remaining neighbors. The triangle centrality for $v$ is,

  \begin{align*}
    \tc(v) &= \frac{1}{3\triangle(G)}\binom{k-1}{2}\Bigl(1+4+k-3\Bigr) \\
    &= \frac{(k+2)}{pk}.
  \end{align*}

  In case (4), a vertex $v$ does not join any $K_k$ but is in the head or tail
  $K_k$. One neighbor joins $K_k$'s. The triangle centrality for $v$ is,

  \begin{align*}
    \tc(v) &= \frac{1}{3\triangle(G)}\binom{k-1}{2}\Bigl(1+2+k-2\Bigr) \\
    &= \frac{(k+1)}{pk}.\qedhere
  \end{align*}
\end{proof}

The next observation follows immediately from
Observation~\ref{obs:clique_chain}.

\begin{observation}
  \label{obs:clique_ring}
  The triangle centrality for a vertex is $\frac{(2k+2)}{pk}$ if it joins copies
  of $K_k$, otherwise it is $\frac{(k+2)}{pk}$, where $G$ is a ring of $p\ge 3$
  copies of $K_k$, each connected by a single vertex, where $k\ge 3$.
\end{observation}

Let us consider the example in Figure~\ref{fig:tc_ex_2} to further demonstrate
the intuition behind Definition~\ref{def:tricent}. There are four 6-cliques
connected to vertex ``a.'' The local triangle count for ``a'' is zero so it
cannot contribute to the centrality of its neighbors, and thus, every vertex in
the 6-cliques must have the same triangle centrality value. Each 6-clique
contains $\binom{6}{3}=20$ triangles leading to $\triangle(G)=80$.

Although each clique vertex in Figure~\ref{fig:tc_ex_2} is locally central
within its clique, there are four such cliques so the overall importance of any
one of the clique vertices should be one-fourth that of a vertex that centered
all triangles. This aligns with Observation~\ref{obs:disjoint_clique} where each
clique vertex has a triangle centrality of $\frac{1}{4}=0.25$.

Now observe that each 6-clique vertex has a local triangle count of
$\binom{5}{2}=10$, which is one-eighth of the triangle count. Although vertex
``a'' is not in any triangles, its four neighbors account for one-half of the
total number of triangles. Thus, we say that vertex ``a'' is the center of half
the triangles in the graph. It follows from Observation~\ref{obs:bridge_clique}
that ``a'' has triangle centrality $\frac{3}{6}=0.5$.

There are some cases where the rank of triangle vertices and their non-triangle
neighbors are the same. Imagine the simple case where the graph has one triangle
and each vertex of that triangle also has neighbors not in the triangle. Then,
we make the following observation.

\begin{observation}
  \label{obs:single}
  The triangle centrality is $1$ for a triangle vertex and its non-triangle
  neighbors if $\triangle(G)=1$.
\end{observation}

The proof of Observation~\ref{obs:single} is obvious. In this case, the triangle
centrality does not discriminate between triangle vertices and their neighbors,
which may be counter-intuitive. Although it can be argued that the triangle
vertices are more important, our goal was to also give importance to vertices
that may not be in triangles. The side effect of this is elucidated by
Observation~\ref{obs:single}.

\section{\label{sec:related}Related Work}
There are many different graph centrality measures, but we will limit our
discussion to closeness~\cite{bib:bavelas1950},
degree~\cite{bib:shaw1954,bib:nieminen1974},
eigenvector~\cite{bib:bonacich1972}, betweenness~\cite{bib:freeman1977}, and
PageRank~\cite{bib:pagerank1998} because these are well-known and are available
in many graph libraries including the Matlab graph toolkit. We refer the reader
to~\cite{bib:rodrigues2018, bib:das2018, bib:networks2010,
  bib:borgatti_everett2006, bib:borgatti2005} for more detail on these common
centrality measures. For convenience, we give a short review of these measures
in Appendix~\ref{sec:review}.

We remark that there is no ``best'' centrality measure. We also stress that the
centrality measures we discuss derive importance based on the network structure,
which may not align with importance due to actual roles and functions of the
real-world entities represented in the graph. The degree of success with
centrality measures therefore relies on how closely the structural network
aligns with the semantic or functional network. Here, we briefly describe and
compare the common centrality measures to triangle centrality in an effort to
show that for some contexts, triangle centrality may be more
appropriate. Moreover, triangle centrality is asymptotically faster than these
other measures with the exception of degree centrality. Recall that
$\bar\delta\ll \sqrt{m}$ in practice so triangle centrality is nearly linear in
runtime on sparse graphs where $m=O(n)$. A summary of the worst-case runtime
(work) for these centralities is given Table~\ref{tbl:work}. Later in
Section~\ref{sec:compare}, we give a more exhaustive comparison of these
centralities on real-world graphs.

\begin{table}[H]
  \caption{\label{tbl:work}Runtime Asymptotic Upper-Bounds}
  \centering
  \begin{tabular}{ll}
    \toprule
    Centrality & Work \\
    \midrule
    Degree & $O(m)$ \\
    Triangle & $O(m\bar\delta)$ \\
    Betweenness & $O(mn)$ \\
    Closeness & $O(mn)$ \\
    PageRank & $O(n^3)$ \\
    Eigenvector & $O(n^3)$ \\
    \bottomrule
  \end{tabular}
\end{table}

Closeness centrality was introduced in 1950 by
Bavelas~\cite{bib:bavelas1950}. It relies on distances and is therefore useful
in gauging longer-range interactions. But a vertex with high degree such as a
star-like subgraph can exert undue influence in this centrality
measure. Closeness centrality also requires finding all-pairs shortest paths and
hence takes $O(mn)$ time, which is much slower than computing triangle
centrality.

Degree centrality was proposed in 1954 by Shaw~\cite{bib:shaw1954} and refined
later in 1974 by Nieminen~\cite{bib:nieminen1974}. It depends only on the
degrees of vertices. Consequently, an e-mail spammer could rank highly in degree
centrality. In contrast, a low-degree vertex can be important in triangle
centrality because its triangle count can be far higher than its degree. The
maximum number of triangles for a vertex is quadratic in its degree, given by
$\binom{d(v)}{2}$. This is compounded if the neighbors of a low-degree vertex
also have close to their maximum triangle count. Such a vertex would be
considered unimportant in degree centrality, which may be misleading in some
contexts as we have argued. For example, vertex ``a'' in
Figure~\ref{fig:tc_ex_2} is ranked last in degree centrality. Then, in
Figure~\ref{fig:tc_ex_1}, all vertices are equally ranked because of uniform
degree, and therefore, no distinction is made among them. Yet clearly the
structure of the network in Figure~\ref{fig:tc_ex_1} is not uniform, and it can
be argued a distinction can be made between the vertices. But degree centrality
is fast to compute, taking only $O(m)$ time. Triangle centrality takes
$O(m\bar\delta)$ time, which in practice is nearly linear-time.

Eigenvector centrality was formalized by Bonacich in
1972~\cite{bib:bonacich1972}. It can be seen as the weighted sum of connections
from any distance in the graph. The central vertex ``a'' in
Figure~\ref{fig:tc_ex_1} is ranked the highest by eigenvector centrality. But
the disadvantage of eigenvector centrality is it is influenced by
degree. Low-degree vertices that act as bridges connecting dense subgraphs would
be ranked poorly by eigenvector centrality, such as vertex ``a'' in
Figure~\ref{fig:tc_ex_2}, but clearly such vertices are important because their
removal would disconnect the subgraphs. The vertex ``a'' in
Figure~\ref{fig:tc_ex_2} is ranked last in eigenvector centrality. It is also
relatively expensive to compute eigenvectors, especially for very large graphs,
making eigenvector centrality less scalable than our triangle
centrality. Eigenvector centrality can be computed in $O(n^3)$ time using the
power iteration method~\cite{bib:golub_vanloan2013}.

Betweenness centrality, introduced in 1977 by Freeman~\cite{bib:freeman1977}, is
similar to closeness centrality in both ranking and runtime. Betweenness
centrality directly accounts for longer range interactions, explicitly relying
on the number of paths through a vertex. Like triangle centrality, it identifies
vertex ``a'' in Figure~\ref{fig:tc_ex_1} and~\ref{fig:tc_ex_2} as being the most
important. But this model of importance relies heavily on distances and does not
capture local subnetwork characteristics. A vertex with many leaf nodes can rank
higher in betweenness centrality than a vertex whose neighbors are
interconnected because the shortest paths between the leaf nodes must go through
their common neighbor. But it can be argued that the vertex whose neighbors are
interconnected is more important. For example, betweenness centrality ranks the
vertex with the most leaf nodes in Figure~\ref{fig:tc_ex_3}
and~\ref{fig:tc_ex_4} as the most important, rather than vertex ``a.'' It is
also more expensive to compute betweenness centrality than triangle centrality
because it requires finding all-pairs shortest paths, taking $O(mn)$
time~\cite{bib:brandes2001}.

PageRank centrality was published in 1998 by Brin and Page as the underlying
technology for the Google search engine~\cite{bib:pagerank1998}. The PageRank
centrality is a variant of eigenvector centrality and therefore has similar
advantages and disadvantages. It is possible that a vertex with many low-quality
connections may still be considered high-ranking in PageRank by nature of having
high degree, making it susceptible to spurious results or gaming. PageRank does
not identify vertex ``a'' as the most central vertex in any of the examples in
Figure~\ref{fig:tc_ex}. Computing the PageRank takes the same time as
eigenvector centrality and hence is slower than computing triangle centrality.

A comparison of the relative ranking of vertex ``a'' in each of the graphs in
Figure~\ref{fig:tc_ex} by the aforementioned centrality measures is given in
Table~\ref{tbl:ex_compare}. We used the graph toolkit in Matlab to compute
rankings from these centrality measures. Only triangle centrality ranks ``a''
first in all Figure~\ref{fig:tc_ex} graphs. It is followed by closeness
centrality which ranks ``a'' first in all but Figure~\ref{fig:tc_ex_4}. The
PageRank centrality does not rank ``a'' first in any of the graphs. The graphs
in Figure~\ref{fig:tc_ex} are idealized but give some insight on the role of
triangles and degree in measuring importance. It is clear that degree plays a
much less significant role for triangle centrality than it does for the other
measures. We consider this is an advantage for triangle centrality because it is
harder to inflate rankings. In Section~\ref{sec:compare}, we compare these
measures on more realistic graphs.

\begin{table}[H]
  \caption{\label{tbl:ex_compare}Rank-Order Comparison for Vertex ``a'' in
    Figure~\ref{fig:tc_ex_1}-\ref{fig:tc_ex_4}}
  \centering
  \begin{tabular}{l *6{S[table-format=2]}}
    \toprule
    {}
    & {Triangle}
    & {Betweenness}
    & {Closeness}
    & {Degree}
    & {Eignevector}
    & {PageRank} \\
    \midrule
    Figure~\ref{fig:tc_ex_1} & 1 & 1 & 1 & 1 & 1 & 7 \\
    Figure~\ref{fig:tc_ex_2} & 1 & 1 & 1 & 25 & 25 & 25 \\
    Figure~\ref{fig:tc_ex_3} & 1 & 2 & 1 & 2 & 2 & 2 \\
    Figure~\ref{fig:tc_ex_4} & 1 & 2 & 2 & 2 & 1 & 2 \\
    \bottomrule
  \end{tabular}
\end{table}

\section{\label{sec:compare}Comparison}
This section compares triangle centrality to the five other centrality measures
discussed in this article using more realistic graphs than those in
Figure~\ref{fig:tc_ex}. Table~\ref{tbl:graphdata} lists the 20 real-world graphs
used in the comparisons. These graphs represent a broad variety of networks to
demonstrate the versatility of triangle centrality. We used Matlab to compute
betweenness, closeness, degree, eigenvector, and PageRank centralities. Prior to
computing the centrality measures, we ensured the graphs were symmetrized,
weights and loops were removed, and vertices were labeled from $1..n$ without
gaps, hence obtaining simple, undirected, and unweighted graphs. There was a
plurality of agreement between triangle centrality and the other measures in
many of these graphs, thus supporting the efficacy of triangle centrality. An
analysis of the results will follow, but first we illustrate four of the smaller
networks to aid in visual comparisons of the centralities. If names of nodes in
these four networks were known, we assigned numeric vertex labels corresponding
to the lexicographic ordering on names. In these next figures, the highest
ranked vertices are indicated by the centrality measures that ranked them.

\begin{table}[H]
  \caption{\label{tbl:graphdata} Test Graphs}
  \sisetup{input-ignore={,},group-separator={,},group-minimum-digits=3}
  \centering
  \footnotesize
  \begin{tabular}{S[table-format=2] l *2{S[table-format=7]} S[table-format=8] r}
  \toprule
  {No.}
  & \multicolumn{1}{c}{Graph}
  & {n (vertices)}
  & {m (edges)}
  & {$\triangle(G)$ (triangles)}
  & \multicolumn{1}{c}{Ref.} \\
  \midrule
  1 & Borgatti 2006 Figure 3
  & 19 & 32 & 13 & \cite{bib:borgatti2006} \\
  2 & Zachary's karate club
  & 34 & 78 & 45 & \cite{bib:zachary1977} \\
  3 & Lusseau's Dolphin network
  & 62 & 159 & 95 & \cite{bib:lusseau2003} \\
  4 & Krebs' 9/11 hijackers network
  & 62 & 153 & 133 & \cite{bib:krebs2002b} \\
  5 & Knuth's Les miserables network
  & 77 & 254 & 467 & \cite{bib:knuth_graphbase} \\
  6 & Krebs' US Political Books network
  & 105 & 441 & 560 & \cite{bib:krebs_polybooks} \\
  7 & Newman's David Copperfield word adjacencies
  & 112 & 425 & 284 & \cite{bib:newman2006} \\
  8 & Girvan-Newman Division IA College Football network
  & 115 & 613 & 810 & \cite{bib:girvan_newman2002} \\
  9 & Watts-Strogatz C. elegans neural network
  & 297 & 2148 & 3241 & \cite{bib:watts_strogatz1998} \\
  10 & Adamic-Glance 2004 US Election Political Blogosphere
  & 1224 & 16,715 & 101,042 & \cite{bib:adamic_glance2005} \\
  11 & Newman's Netscience Co-Authorship network
  & 1461 & 2742 & 3764 & \cite{bib:newman2006} \\
  12 & Watts-Strogatz Western States Power Grid
  & 4941 & 6594 & 651 & \cite{bib:watts_strogatz1998} \\
  13 & SNAP ca-HepTh & 9875 & 25,973 & 28,339 & \cite{bib:snapnets} \\
  14 & SNAP ca-AstroPh & 18,771 & 198,050 & 1,351,441 & \cite{bib:snapnets} \\
  15 & SNAP e-mail-Enron & 36,692 & 183,831 & 727,044 & \cite{bib:snapnets} \\
  16 & Newman's Condensed Matter Physics network
  & 39,577 & 175,692 & 378,063 & \cite{bib:newman2001} \\
  17 & SNAP web-Stanford & 281,903 & 1,992,636 & 11,329,473 & \cite{bib:snapnets} \\
  18 & SNAP com-DBLP & 317,080 & 1,049,866 & 2,224,385 & \cite{bib:snapnets} \\
  19 & SNAP com-Amazon & 334,863 & 925,872 & 667,129 & \cite{bib:snapnets} \\
  20 & SNAP roadNet-PA & 1,088,092 & 1,541,898 & 67,150 & \cite{bib:snapnets} \\
  \bottomrule
  \end{tabular}
\end{table}

Figure~\ref{fig:borgatti} depicts a network that appeared in Borgatti's 2006
article~\cite[Figure~3]{bib:borgatti2006}. This network is interesting because 4
of the 19 vertices, more than 20\%, are considered central according to the
centrality measures described in this article. There appear to be two clusters
that can be disconnected if the vertex ranked highest by betweenness or
closeness is removed. But there is a greater agreement among the centrality
measures on the vertex ranked highest by triangle centrality.

\begin{figure}[H]
\centering
  \begin{tikzpicture}
    [scale=0.5, inner sep=1.5, minimum size=14,
    vertex/.style={circle,draw=black,thin},
    edge/.style={thin}]

    \node [vertex] (a) {a};
    \node [vertex] (b) at ([shift={(70:2.85)}]a) {b};
    \node [vertex, fill=gray!30] (d) at ([shift={(25:1.5)}]a) {d};
    \node [vertex] (f) at ([shift={(335:2.5)}]a) {f};
    \node [vertex] (e) at ([shift={(15:1.5)}]b) {e};
    \node [vertex] (c) at ([shift={(275:3.15)}]d) {c};
    \node [vertex] (g) at ([shift={(5:1.75)}]d) {g};
    \node [vertex, fill=gray!30] (h) at ([shift={(5:2.15)}]c) {h};
    \node [vertex, fill=gray!30] (i) at ([shift={(350:2.45)}]h) {i};
    \node [vertex, fill=gray!30] (j) at ([shift={(358:2.5)}]i) {j};
    \node [vertex] (p) at ([shift={(320:2.25)}]i) {p};
    \node [vertex] (n) at ([shift={(290:3.00)}]i) {n};
    \node [vertex] (k) at ([shift={(310:2.25)}]j) {k};
    \node [vertex] (l) at ([shift={(10:2.15)}]j) {l};
    \node [vertex] (m) at ([shift={(60:3.05)}]j) {m};
    \node [vertex] (o) at ([shift={(275:3.75)}]j) {o};
    \node [vertex] (q) at ([shift={(45:2.5)}]m) {q};
    \node [vertex] (s) at ([shift={(30:2.15)}]q) {s};
    \node [vertex] (r) at ([shift={(260:3.5)}]h) {r};

    \draw [edge] (a) -- (b);
    \draw [edge] (a) -- (d);
    \draw [edge] (a) -- (f);
    \draw [edge] (b) -- (e);
    \draw [edge] (b) -- (g);
    \draw [edge] (b) -- (d);
    \draw [edge] (c) -- (h);
    \draw [edge] (d) -- (c);
    \draw [edge] (d) -- (e);
    \draw [edge] (d) -- (g);
    \draw [edge] (d) -- (f);
    \draw [edge] (e) -- (g);
    \draw [edge] (f) -- (g);
    \draw [edge] (f) -- (h);
    \draw [edge] (g) -- (h);
    \draw [edge] (h) -- (i);
    \draw [edge] (h) -- (r);
    \draw [edge] (i) -- (j);
    \draw [edge] (i) -- (p);
    \draw [edge] (i) -- (n);
    \draw [edge] (j) -- (k);
    \draw [edge] (j) -- (l);
    \draw [edge] (j) -- (m);
    \draw [edge] (j) -- (o);
    \draw [edge] (j) -- (p);
    \draw [edge] (k) -- (p);
    \draw [edge] (k) -- (l);
    \draw [edge] (l) -- (m);
    \draw [edge] (p) -- (n);
    \draw [edge] (n) -- (o);
    \draw [edge] (m) -- (q);
    \draw [edge] (q) -- (s);

    \matrix (m) [
      right=2cm of m,
      matrix of nodes,
      column sep=2pt,
      nodes={inner sep=1, minimum size=0}
    ] {
      \node[vertex, fill=gray!30, inner sep=.75, minimum size=15] {d};
      & \node[vertex, fill=gray!30, inner sep=.75, minimum size=15] {h};
      & \node[vertex, fill=gray!30, inner sep=.75, minimum size=15] {i};
      & \node[vertex, fill=gray!30, inner sep=.75, minimum size=15] {j}; \\
      \midrule
      \node {\tc}; & \node {\bc}; & \node {\cc}; & \node {\dc}; \\
      \node {\dc}; & {} & {} & \node {\pr}; \\
      \node {\ev}; & {} & {} & {} \\
    };
  \end{tikzpicture}
\caption{\label{fig:borgatti}Borgatti 2006~\cite[Fig.~3]{bib:borgatti2006}.}
\end{figure}

Figure~\ref{fig:karate} depicts a similar illustration of Zachary's karate club
social network~\cite{bib:zachary1977} that appears
in~\cite[Figure~2]{bib:whang2015}. This is a well-studied network and serves as
a benchmark for clustering and community
detection~\cite{bib:girvan_newman2002}. The network was constructed by
Zachary~\cite{bib:zachary1977} after observing 34 members in a karate club from
1970 to 1972. Following a dispute between the instructor (vertex~$1$) and
administrator (vertex~$34$), the karate club network split into two respective
communities~\cite{bib:girvan_newman2002}. The instructor and administrator were
ranked highest by all the centrality measures except for triangle
centrality. The triangle centrality is alone in ranking vertex~$14$ as the most
central. Vertex degree plays a significant role in this network. The instructor
and administrator have the two highest degrees in the graph, respectively, 16
and 17. In contrast, vertex~$14$ has degree 5. We remind the reader that the
rankings rely on the network structure. While the functional roles of the karate
club's two main adversaries imply their influence, it may be that vertex~$14$ is
more important from a structural standpoint. In addition to vertex~$14$ being
central with respect to triangles, it also connects the two communities and
appears in the overlap between them~\cite{bib:zarei2009, bib:whang2015,
  bib:lierde2020}.

\begin{figure}[H]
  \centering
  \begin{tikzpicture}
    [scale=0.7, inner sep=.75, minimum size=14,
    vertex/.style={circle,draw=black,thin},
    edge/.style={thin}]

    \node [vertex,fill=gray!30] (14) {14};
    \node [vertex,fill=gray!30] (1) at ([shift={(345:4)}]14) {1};
    \node [vertex,fill=gray!30] (34) at ([shift={(180:4)}]14) {34};

    \node [vertex] (2) at ([shift={(245:2.5)}]1) {2};
    \node [vertex] (3) at ([shift={(195:4)}]1) {3};
    \node [vertex] (4) at ([shift={(110:2.5)}]1) {4};
    \node [vertex] (5) at ([shift={(310:2.5)}]1) {5};
    \node [vertex] (6) at ([shift={(15:2.75)}]1) {6};
    \node [vertex] (7) at ([shift={(345:2.75)}]1) {7};
    \node [vertex] (8) at ([shift={(230:3.5)}]1) {8};
    \node [vertex] (9) at ([shift={(185:5)}]1) {9};
    \node [vertex] (11) at ([shift={(50:2.5)}]1) {11};
    \node [vertex] (12) at ([shift={(0:2.15)}]1) {12};
    \node [vertex] (13) at ([shift={(80:2)}]1) {13};
    \node [vertex] (18) at ([shift={(270:1.5)}]1) {18};
    \node [vertex] (20) at ([shift={(150:4.75)}]1) {20};
    \node [vertex] (22) at ([shift={(290:2.25)}]1) {22};

    \node [vertex] (10) at ([shift={(10:3)}]34) {10};
    \node [vertex] (15) at ([shift={(120:2)}]34) {15};
    \node [vertex] (16) at ([shift={(160:2.15)}]34) {16};
    \node [vertex] (19) at ([shift={(205:4.25)}]34) {19};
    \node [vertex] (21) at ([shift={(190:4)}]34) {21};
    \node [vertex] (23) at ([shift={(60:2.25)}]34) {23};
    \node [vertex] (24) at ([shift={(220:4)}]34) {24};
    \node [vertex] (27) at ([shift={(140:2.5)}]34) {27};
    \node [vertex] (28) at ([shift={(250:2.75)}]34) {28};
    \node [vertex] (29) at ([shift={(315:4)}]34) {29};
    \node [vertex] (30) at ([shift={(175:3.25)}]34) {30};
    \node [vertex] (31) at ([shift={(350:2.5)}]34) {31};
    \node [vertex] (32) at ([shift={(295:4.25)}]34) {32};
    \node [vertex] (33) at ([shift={(270:1.75)}]34) {33};

    \node [vertex] (17) at ([shift={(320:1)}]6) {17};
    \node [vertex] (25) at ([shift={(340:1.25)}]28) {25};
    \node [vertex] (26) at ([shift={(320:1.5)}]24) {26};

    \draw [edge] (14) -- (1);
    \draw [edge] (14) -- (2);
    \draw [edge] (14) -- (3);
    \draw [edge] (14) -- (4);
    \draw [edge] (14) -- (34);

    \draw [edge] (1) -- (2);
    \draw [edge] (1) -- (3);
    \draw [edge] (1) -- (4);
    \draw [edge] (1) -- (5);
    \draw [edge] (1) -- (6);
    \draw [edge] (1) -- (7);
    \draw [edge] (1) -- (8);
    \draw [edge] (1) -- (9);
    \draw [edge] (1) -- (11);
    \draw [edge] (1) -- (12);
    \draw [edge] (1) -- (13);
    \draw [edge] (1) -- (18);
    \draw [edge] (1) -- (20);
    \draw [edge] (1) -- (22);
    \draw [edge] (1) -- (32);

    \draw [edge] (34) -- (9);
    \draw [edge] (34) -- (10);
    \draw [edge] (34) -- (15);
    \draw [edge] (34) -- (16);
    \draw [edge] (34) -- (19);
    \draw [edge] (34) -- (20);
    \draw [edge] (34) -- (21);
    \draw [edge] (34) -- (23);
    \draw [edge] (34) -- (24);
    \draw [edge] (34) -- (27);
    \draw [edge] (34) -- (28);
    \draw [edge] (34) -- (29);
    \draw [edge] (34) -- (30);
    \draw [edge] (34) -- (31);
    \draw [edge] (34) -- (32);
    \draw [edge] (34) -- (33);

    \draw [edge] (2) -- (3);
    \draw [edge] (2) -- (4);
    \draw [edge] (2) -- (8);
    \draw [edge] (2) -- (18);
    \draw [edge] (2) -- (20);
    \draw [edge] (2) -- (22);
    \draw [edge] (2) -- (31);
    \draw [edge] (3) -- (4);
    \draw [edge] (3) -- (8);
    \draw [edge] (3) -- (9);
    \draw [edge] (3) -- (10);
    \draw [edge] (3) -- (28);
    \draw [edge] (3) -- (29);
    \draw [edge] (3) -- (33);
    \draw [edge] (4) -- (8);
    \draw [edge] (4) -- (13);
    \draw [edge] (5) -- (7);
    \draw [edge] (5) -- (11);
    \draw [edge] (6) -- (7);
    \draw [edge] (6) -- (11);
    \draw [edge] (6) -- (17);
    \draw [edge] (7) -- (17);
    \draw [edge] (9) -- (31);
    \draw [edge] (15) -- (33);
    \draw [edge] (16) -- (33);
    \draw [edge] (19) -- (33);
    \draw [edge] (21) -- (33);
    \draw [edge] (23) -- (33);
    \draw [edge] (24) -- (26);
    \draw [edge] (24) -- (28);
    \draw [edge] (24) -- (30);
    \draw [edge] (24) -- (33);
    \draw [edge] (25) -- (26);
    \draw [edge] (25) -- (28);
    \draw [edge] (25) -- (32);
    \draw [edge] (26) -- (32);
    \draw [edge] (27) -- (30);
    \draw [edge] (29) -- (32);
    \draw [edge] (30) -- (33);
    \draw [edge] (32) -- (33);

    \matrix (m) [
      right=3cm of 13,
      matrix of nodes,
      column sep=2pt,
      nodes={inner sep=1, minimum size=0}
    ] {
      \node[vertex, fill=gray!30, inner sep=.75, minimum size=15] {1};
      & \node[vertex, fill=gray!30, inner sep=.75, minimum size=15] {14};
      & \node[vertex, fill=gray!30, inner sep=.75, minimum size=15] {34}; \\
      \midrule
      \node {\bc}; & \node {\tc}; & \node {\dc}; \\
      \node {\cc}; & {} & \node {\ev}; \\
      {} & {} & \node {\pr}; \\
    };

    \draw[dashed] (-3.5, -1) ellipse (5.5cm and 3.75cm);
    \draw[dashed] (2.5, -1) ellipse (5.5cm and 3.75cm);
  \end{tikzpicture}
  \caption{\label{fig:karate}Zachary's karate club.}
\end{figure}

Figure~\ref{fig:dolphin} depicts Lusseau's social network of 62 bottlenose
dolphins living off Doubtful Sound in New Zealand between 1994 and
2001~\cite{bib:lusseau2003}. This is another benchmark network for
clustering. According to Lusseau and Newman~\cite{bib:lusseau_newman2004}, the
disappearance of one dolphin, named SN100 (vertex~$37$), split the network into
two communities, but when SN100 reappeared the communities rejoined. Then, it
comes as no surprise that SN100 is considered the most central by the
betweenness and closeness centralities. But it is the dolphin named Grin
(vertex~$15$) that is ranked highest by the remaining centralities including
triangle centrality. Grin also has the highest degree in the graph.

\begin{figure}[H]
  \centering
  \begin{tikzpicture}
    [scale=0.55, inner sep=.75, minimum size=15,
    vertex/.style={circle,draw=black,thin},
    edge/.style={thin},
    tname/.style args={#1:#2}{%
      label={[label distance=0,fill=white,minimum height=0,font=\tiny]#1:#2}
    },
    tname/.default={#1:#2:-4}]

    \node [vertex,tname=350:Beak] (1) {1};
    \node [vertex,tname=270:Fish] (11) at ([shift={(10:3.5)}]1) {11};
    \node [vertex,tname=90:Grin,fill=gray!30] (15) at ([shift={(70:8)}]1) {15};
    \node [vertex,tname=below:SN96] (43) at ([shift={(300:3.5)}]1) {43};
    \node [vertex,tname=below:TR77] (48) at ([shift={(265:2)}]1) {48};
    \node [vertex,tname=10:CCL] (4) at ([shift={(25:2.5)}]1) {4};
    \node [vertex,tname=80:Beescratch] (2) at ([shift={(175:5.5)}]1) {2};
    \node [vertex,tname=below:Bumper] (3) at ([shift={(340:5.5)}]1) {3};

    \node [vertex,tname=270:Jet] (18) at ([shift={(210:3)}]2) {18};
    \node [vertex,tname=below:Knit] (20) at ([shift={(290:3.75)}]2) {20};
    \node [vertex,tname=below:Notch] (27) at ([shift={(270:3.75)}]2) {27};
    \node [vertex,tname=above:Number1] (28) at ([shift={(245:2.5)}]2) {28};
    \node [vertex,tname=0:Oscar] (29) at ([shift={(340:3.5)}]2) {29};
    \node [vertex,tname=90:SN100,fill=gray!30] (37) at ([shift={(90:1.5)}]2) {37};
    \node [vertex,tname=above:SN90] (42) at ([shift={(155:2.75)}]2) {42};
    \node [vertex,tname=50:Upbang] (55) at ([shift={(185:2.5)}]2) {55};
    \node [vertex,tname=above:DN16] (6) at ([shift={(175:7)}]2) {6};
    \node [vertex,tname=above:DN63] (8) at ([shift={(320:2)}]2) {8};

    \node [vertex,tname=90:Thumper] (45) at ([shift={(95:3)}]3) {45};
    \node [vertex,tname=85:Zipfel] (62) at ([shift={(165:1.5)}]3) {62};
    \node [vertex,tname=275:TSN83] (54) at ([shift={(70:2)}]3) {54};

    \node [vertex,tname=above:Double] (9) at ([shift={(180:5)}]4) {9};

    \node [vertex,tname=left:Feather] (10) at ([shift={(300:1.5)}]6) {10};
    \node [vertex,tname=above:Gallatin] (14) at ([shift={(340:3)}]6) {14};
    \node [vertex,tname=above:Wave] (57) at ([shift={(220:2.5)}]6) {57};
    \node [vertex,tname=100:Web] (58) at ([shift={(20:2.5)}]6) {58};
    \node [vertex,tname=below:DN21] (7) at ([shift={(300:3.25)}]6) {7};
    \node [vertex,tname=above:TR82] (49) at ([shift={(70:2.5)}]6) {49};

    \node [vertex,tname=below:PL] (31) at ([shift={(310:3)}]8) {31};

    \node [vertex,tname=above:Kringel] (21) at ([shift={(15:4.5)}]9) {21};
    \node [vertex,tname=95:SN4] (38) at ([shift={(50:2.5)}]9) {38};
    \node [vertex,tname=above:Topless] (46) at ([shift={(95:6.5)}]9) {46};

    \node [vertex,tname=below:Ripplefluke] (33) at ([shift={(240:2)}]10) {33};

    \node [vertex,tname=above:Fork] (13) at ([shift={(80:8.5)}]11) {13};

    \node [vertex,tname=5:Hook] (17) at ([shift={(290:5)}]15) {17};
    \node [vertex,tname=above:MN83] (25) at ([shift={(135:5.5)}]15) {25};
    \node [vertex,tname=290:Shmuddel] (35) at ([shift={(325:6.25)}]15) {35};
    \node [vertex,tname=0:SN63] (39) at ([shift={(310:7)}]15) {39};
    \node [vertex,tname=80:Stripes] (44) at ([shift={(350:3)}]15) {44};
    \node [vertex,tname=90:TR99] (51) at ([shift={(215:3.75)}]15) {51};
    \node [vertex,tname=above:TSN103] (53) at ([shift={(85:2.5)}]15) {53};
    \node [vertex,tname=80:Jonah] (19) at ([shift={(175:2)}]15) {19};
    \node [vertex,tname=5:Scabs] (34) at ([shift={(280:3.5)}]15) {34};

    \node [vertex,tname=below:MN23] (23) at ([shift={(225:2.5)}]18) {23};
    \node [vertex,tname=below:Mus] (26) at ([shift={(290:2.5)}]18) {26};
    \node [vertex,tname=below:Quasi] (32) at ([shift={(255:3)}]18) {32};

    \node [vertex,tname=90:MN105] (22) at ([shift={(165:2.25)}]19) {22};
    \node [vertex,tname=left:Trigger] (52) at ([shift={(150:5)}]19) {52};
    \node [vertex,tname=above:Patchback] (30) at ([shift={(95:3)}]19) {30};

    \node [vertex,tname=above:SMN5] (36) at ([shift={(20:2)}]30) {36};

    \node [vertex,tname=above:Zig] (61) at ([shift={(200:2)}]33) {61};

    \node [vertex,tname=above:TR88] (50) at ([shift={(40:1.5)}]35) {50};

    \node [vertex,tname=90:MN60] (24) at ([shift={(95:6)}]37) {24};
    \node [vertex,tname=above:SN89] (40) at ([shift={(160:4)}]37) {40};
    \node [vertex,tname=95:SN9] (41) at ([shift={(70:3.5)}]37) {41};
    \node [vertex,tname=above left:Zap] (60) at ([shift={(125:2.5)}]37) {60};

    \node [vertex,tname=275:Whitetip] (59) at ([shift={(270:1.5)}]39) {59};

    \node [vertex,tname=above right:TR120] (47) at ([shift={(310:2)}]44) {47};

    \node [vertex,tname=above:Cross] (5) at ([shift={(145:2)}]52) {5};
    \node [vertex,tname=above:Five] (12) at ([shift={(90:1.5)}]52) {12};
    \node [vertex,tname=above:Vau] (56) at ([shift={(210:5)}]52) {56};

    \node [vertex,tname=left:Haecksel] (16) at ([shift={(100:2)}]60) {16};

    \begin{pgfonlayer}{background}
    \draw (1) -- (11);
    \draw (1) -- (15);
    \draw (1) -- (16);
    \draw (1) -- (41);
    \draw (1) -- (43);
    \draw (1) -- (48);

    \draw (2) -- (18);
    \draw (2) -- (20);
    \draw (2) -- (27);
    \draw (2) -- (28);
    \draw (2) -- (29);
    \draw (2) -- (37);
    \draw (2) -- (42);
    \draw (2) -- (55);

    \draw (3) -- (11);
    \draw (3) -- (43);
    \draw (3) -- (45);
    \draw (3) -- (62);

    \draw (4) -- (9);
    \draw (4) -- (15);
    \draw (4) -- (60);

    \draw (5) -- (52);

    \draw (6) -- (10);
    \draw (6) -- (14);
    \draw (6) -- (57);
    \draw (6) -- (58);

    \draw (7) -- (10);
    \draw (7) -- (14);
    \draw (7) -- (18);
    \draw (7) -- (55);
    \draw (7) -- (57);
    \draw (7) -- (58);

    \draw (8) -- (20);
    \draw (8) -- (28);
    \draw (8) -- (31);
    \draw (8) -- (41);
    \draw (8) -- (55);

    \draw (9) -- (21);
    \draw (9) -- (29);
    \draw (9) -- (38);
    \draw (9) -- (46);
    \draw (9) -- (60);

    \draw (10) -- (14);
    \draw (10) -- (18);
    \draw (10) -- (33);
    \draw (10) -- (42);
    \draw (10) -- (58);

    \draw (11) -- (30);
    \draw (11) -- (43);
    \draw (11) -- (48);

    \draw (12) -- (52);

    \draw (13) -- (34);

    \draw (14) -- (18);
    \draw (14) -- (33);
    \draw (14) -- (42);
    \draw (14) -- (55);
    \draw (14) -- (58);

    \draw (15) -- (17);
    \draw (15) -- (25);
    \draw (15) -- (34);
    \draw (15) -- (38);
    \draw (15) -- (39);
    \draw (15) -- (41);
    \draw (15) -- (44);
    \draw (15) -- (51);
    \draw (15) -- (53);

    \draw (16) -- (19);
    \draw (16) -- (25);
    \draw (16) -- (41);
    \draw (16) -- (46);
    \draw (16) -- (56);
    \draw (16) -- (60);

    \draw (17) -- (21);
    \draw (17) -- (34);
    \draw (17) -- (38);
    \draw (17) -- (39);
    \draw (17) -- (51);

    \draw (18) -- (23);
    \draw (18) -- (26);
    \draw (18) -- (28);
    \draw (18) -- (32);
    \draw (18) -- (58);

    \draw (19) -- (21);
    \draw (19) -- (22);
    \draw (19) -- (25);
    \draw (19) -- (30);
    \draw (19) -- (46);
    \draw (19) -- (52);

    \draw (20) -- (31);
    \draw (20) -- (55);

    \draw (21) -- (29);
    \draw (21) -- (37);
    \draw (21) -- (39);
    \draw (21) -- (45);
    \draw (21) -- (48);
    \draw (21) -- (51);

    \draw (22) -- (30);
    \draw (22) -- (34);
    \draw (22) -- (38);
    \draw (22) -- (46);
    \draw (22) -- (52);

    \draw (24) -- (37);
    \draw (24) -- (46);
    \draw (24) -- (52);

    \draw (25) -- (30);
    \draw (25) -- (46);
    \draw (25) -- (52);

    \draw (26) -- (27);
    \draw (26) -- (28);

    \draw (27) -- (28);

    \draw (29) -- (31);
    \draw (29) -- (48);

    \draw (30) -- (36);
    \draw (30) -- (44);
    \draw (30) -- (46);
    \draw (30) -- (52);
    \draw (30) -- (53);

    \draw (31) -- (43);
    \draw (31) -- (48);

    \draw (33) -- (61);

    \draw (34) -- (35);
    \draw (34) -- (38);
    \draw (34) -- (39);
    \draw (34) -- (41);
    \draw (34) -- (44);
    \draw (34) -- (51);

    \draw (35) -- (38);
    \draw (35) -- (45);
    \draw (35) -- (50);

    \draw (37) -- (38);
    \draw (37) -- (40);
    \draw (37) -- (41);
    \draw (37) -- (60);

    \draw (38) -- (41);
    \draw (38) -- (44);
    \draw (38) -- (46);
    \draw (38) -- (62);

    \draw (39) -- (44);
    \draw (39) -- (45);
    \draw (39) -- (53);
    \draw (39) -- (59);

    \draw (40) -- (58);

    \draw (41) -- (53);

    \draw (42) -- (55);
    \draw (42) -- (58);

    \draw (43) -- (48);
    \draw (43) -- (51);

    \draw (44) -- (47);
    \draw (44) -- (54);

    \draw (46) -- (51);
    \draw (46) -- (52);
    \draw (46) -- (60);

    \draw (47) -- (50);

    \draw (49) -- (58);

    \draw (51) -- (52);

    \draw (52) -- (56);

    \draw (54) -- (62);

    \draw (55) -- (58);
    \end{pgfonlayer}

    \matrix (m) [
      right=3.5cm of 15,
      matrix of nodes,
      column sep=2pt,
      nodes={inner sep=1, minimum size=0}
    ] {
      \node[vertex, fill=gray!30] {15}; & \node[vertex,fill=gray!30] {37}; \\
      \midrule
      \node {\tc}; & \node {\bc}; \\
      \node {\dc}; & \node {\cc}; \\
      \node {\ev}; & {} \\
      \node {\pr}; & {} \\
    };
  \end{tikzpicture}
  \caption{\label{fig:dolphin}Lusseau's Dolphin network.}
\end{figure}

Figure~\ref{fig:911_hijackers} depicts Krebs' network of the 9/11 hijackers and
their accomplices~\cite{bib:krebs2002, bib:krebs2002b}. This network of 62
vertices includes the 19 hijackers and 43 co-conspirators. The functional leader
of the hijackers was Mohamed Atta (vertex 38), and it is evident from the
network that he played an important structural role. All centrality measures in
this study ranked Mohamed Atta the highest. This is an example showing that
triangle centrality aligns with the consensus on a central node. We also note
that Marwan Al-Shehhi is ranked the second highest by triangle centrality, in
agreement with closeness, degree, and eigenvector centralities. Thus, triangle
centrality is with the majority in ranking the top two vertices in this network.

Mohamed Atta has the highest degree and second-highest triangle count,
respectively, 22 and 42. Marwan Al-Shehhi (vertex 35) has the second-highest
degree and highest triangle count, respectively, 18 and 47. The largest clique
size in the graph is six, and one such clique (vertices 2, 20, 35, 38, 55, and
58) contains both Mohamed Atta and Marwan Al-Shehhi, but Marwan Al-Shehhi is
also a member of an overlapping 6-clique (vertices 2, 20, 35, 55, 58, and
59). These two highly ranked vertices gain the same triangle contribution from
14 common neighbors, but Mohamed Atta is at the center of more triangles.

\begin{figure}[H]
  \centering
  \begin{tikzpicture}
    [scale=0.5, inner sep=.75, minimum size=15,
    vertex/.style={circle,draw=black,thin},
    edge/.style={thin},
    tname/.style args={#1:#2}{%
      label={[label distance=0,fill=white,minimum height=0,font=\tiny]#1:#2}
    },
    tname/.default={#1:#2:-4}]

    \node [vertex,tname=270:Mohamed Atta,fill=gray!30] (38) {38};
    \node [vertex,tname=350:Abdelghani Mzoudi] (1) at ([shift={(340:6)}]38) {1};
    \node [vertex,tname=0:{\parbox[t]{1.25cm}{Abdul Aziz\\Al-Omari*}}] (2)
    at ([shift={(100:8.5)}]38) {2};
    \node [vertex,tname=10:Agus Budiman] (7) at ([shift={(20:5)}]38) {7};
    \node [vertex,tname=0:Ahmed Al Haznawi] (9) at ([shift={(135:8.5)}]38) {9};
    \node [vertex,
    tname=190:{\parbox[t]{1.5cm}{Ahmed Khalil\\Ibrahim Samir Al-Ani}}] (11)
    at ([shift={(195:5.5)}]38) {11};
    \node [vertex,tname=270:{\parbox[t]{1.5cm}{Essid Sami\\Ben Khemais}}] (16)
    at ([shift={(225:10)}]38) {16};
    \node [vertex,tname=0:Fayez Ahmed] (20) at ([shift={(45:6)}]38) {20};
    \node [vertex,tname=270:Hani Hanjour] (22) at ([shift={(170:8.5)}]38) {22};
    \node [vertex,tname=80:{\parbox[t]{1.75cm}{Imad Eddin\\Barakat Yarkas}}] (24)
    at ([shift={(285:6)}]38) {24};
    \node [vertex,tname=190:Lofti Raissi] (30) at ([shift={(180:5)}]38) {30};
    \node [vertex,tname=270:{\parbox[t]{1.15cm}{Mamoun\\Darkazanli}}] (34)
    at ([shift={(210:5.5)}]38) {34};
    \node [vertex,tname=270:Marwan Al-Shehhi] (35) at ([shift={(115:5.5)}]38)
    {35};
    \node [vertex,tname=0:Mounir El Motassadeq] (42) at ([shift={(325:6)}]38)
    {42};
    \node [vertex,tname=0:{\parbox[t]{1.75cm}{Mustafa Ahmed\\al-Hisawi}}] (43)
    at ([shift={(30:6)}]38) {43};
    \node [vertex,tname=270:Nawaf Alhazmi] (45) at ([shift={(160:10.5)}]38) {45};
    \node [vertex,tname=0:Ramzi Bin al-Shibh] (49) at ([shift={(0:4.5)}]38)
    {49};
    \node [vertex,tname=270:Said Bahaji] (52) at ([shift={(265:4.5)}]38)
    {52};
    \node [vertex,tname=0:Satam Suqami] (55) at ([shift={(80:9.75)}]38) {55};
    \node [vertex,tname=0:Wail Alshehri] (58) at ([shift={(60:9.5)}]38) {58};
    \node [vertex,tname=270:Zacarias Moussaoui] (60) at ([shift={(300:7.25)}]38)
    {60};
    \node [vertex,tname=190:Zakariya Essabar] (61) at ([shift={(240:5)}]38)
    {61};
    \node [vertex,tname=200:Ziad Jarrah] (62) at ([shift={(150:5.5)}]38) {62};

    \node [vertex,tname=190:Abdussattar Shaikh] (3) at ([shift={(180:3.5)}]45)
    {3};
    \node [vertex,tname=270:Hamza Alghamdi] (21) at ([shift={(65:2)}]45) {21};
    \node [vertex,tname=180:Khalid Al-Mihdhar] (28) at ([shift={(210:2.5)}]45)
    {28};
    \node [vertex,tname=180:Mohamed Abdi] (37) at ([shift={(135:3)}]45) {37};
    \node [vertex,tname=180:Osama Awadallah] (47) at ([shift={(155:2.5)}]45)
    {47};
    \node [vertex,tname=270:Salem Alhazmi*] (53) at ([shift={(10:3)}]45) {53};

    \node [vertex,tname=270:Ahmed Alghamdi] (8) at ([shift={(60:3.5)}]21) {8};
    \node [vertex,tname=180:Ahmed Alnami] (10) at ([shift={(150:2.5)}]21) {10};
    \node [vertex,tname=0:Mohand Alshehri*] (41) at ([shift={(25:4.5)}]21)
    {41};
    \node [vertex,tname=180:Saeed Alghamdi*] (51) at ([shift={(115:3)}]21) {51};

    \node [vertex,tname=190:Bandar Alhazmi] (13) at ([shift={(250:4)}]22) {13};
    \node [vertex,tname=190:Faisal Al Salmi] (19) at ([shift={(205:3.5)}]22)
    {19};
    \node [vertex,tname=190:Majed Moqed] (32) at ([shift={(185:3)}]22) {32};
    \node [vertex,tname=180:{\parbox[t]{1.2cm}{Rayed\\Mohammed Abdullah}}] (50)
    at ([shift={(225:4)}]22) {50};

    \node [vertex,tname=190:Mamduh Mahmud Salim] (33) at ([shift={(200:3)}]34)
    {33};

    \node [vertex,tname=180:Nabil al-Marabh] (44) at ([shift={(60:1.75)}]51) {44};
    \node [vertex,tname=270:Raed Hijazi] (48) at ([shift={(20:8)}]51) {48};

    \node [vertex,tname=0:Waleed Alshehri] (59) at ([shift={(75:3)}]43)
    {59};

    \node [vertex,tname=350:Essoussi Laaroussi] (17) at ([shift={(335:3)}]16)
    {17};
    \node [vertex,tname=180:Fahid al Shakri] (18) at ([shift={(190:4)}]16) {18};
    \node [vertex,tname=300:Haydar Abu Doha] (23) at ([shift={(295:3.5)}]16) {23};
    \node [vertex,tname=180:Lased Ben Heni] (29) at ([shift={(175:3)}]16) {29};
    \node [vertex,tname=180:Madjid Sahoune] (31) at ([shift={(210:4)}]16) {31};
    \node [vertex,tname=180:Mehdi Khammoun] (36) at ([shift={(230:4)}]16) {36};
    \node [vertex,tname=350:Mohamed Bensakhria] (39) at ([shift={(270:2)}]16)
    {39};
    \node [vertex,tname=180:Samir Kishk] (54) at ([shift={(145:2)}]16) {54};
    \node [vertex,tname=180:Seifallah ben Hassine] (56) at ([shift={(100:2)}]16)
    {56};
    \node [vertex,tname=270:Tarek Maaroufi] (57) at ([shift={(10:4)}]16) {57};

    \node [vertex,tname=0:Mohammed Belfas] (40) at ([shift={(335:2)}]7) {40};

    \node [vertex,tname=180:Abu Qatada] (4) at ([shift={(205:3.5)}]60) {4};
    \node [vertex,tname=270:Ahmed Ressam] (12) at ([shift={(230:3.5)}]60) {12};
    \node [vertex,tname=270:David Courtaillier] (14) at ([shift={(250:4.75)}]60)
    {14};
    \node [vertex,tname=0:Djamal Beghal] (15) at ([shift={(330:3.75)}]60) {15};
    \node [vertex,tname=0:Jerome Courtaillier] (26) at ([shift={(275:4.5)}]60)
    {26};
    \node [vertex,tname=0:Kamel Daoudi] (27) at ([shift={(35:2.75)}]60) {27};

    \node [vertex,tname=0:Abu Walid] (5) at ([shift={(250:2)}]15) {5};
    \node [vertex,tname=320:Abu Zubeida] (6) at ([shift={(330:2)}]15) {6};
    \node [vertex,tname=0:{\parbox[t]{1.25cm}{Jean-Marc\\Grandvisir}}] (25)
    at ([shift={(70:2.5)}]15) {25};
    \node [vertex,tname=350:Nizar Trabelsi] (46) at ([shift={(30:2)}]15) {46};

    \begin{pgfonlayer}{background}
    \draw (1) -- (38);
    \draw (2) -- (8);
    \draw (2) -- (20);
    \draw (2) -- (22);
    \draw (2) -- (35);
    \draw (2) -- (38);
    \draw (2) -- (53);
    \draw (2) -- (55);
    \draw (2) -- (58);
    \draw (2) -- (59);
    \draw (3) -- (28);
    \draw (3) -- (45);
    \draw (3) -- (47);
    \draw (4) -- (5);
    \draw (4) -- (15);
    \draw (4) -- (24);
    \draw (4) -- (57);
    \draw (4) -- (60);
    \draw (5) -- (15);
    \draw (5) -- (27);
    \draw (6) -- (15);
    \draw (7) -- (35);
    \draw (7) -- (38);
    \draw (7) -- (40);
    \draw (7) -- (49);
    \draw (7) -- (62);
    \draw (8) -- (21);
    \draw (8) -- (22);
    \draw (8) -- (44);
    \draw (8) -- (53);
    \draw (9) -- (21);
    \draw (9) -- (38);
    \draw (9) -- (51);
    \draw (9) -- (62);
    \draw (10) -- (21);
    \draw (10) -- (45);
    \draw (10) -- (51);
    \draw (11) -- (38);
    \draw (12) -- (23);
    \draw (12) -- (60);
    \draw (13) -- (22);
    \draw (13) -- (50);
    \draw (14) -- (26);
    \draw (14) -- (60);
    \draw (15) -- (25);
    \draw (15) -- (26);
    \draw (15) -- (27);
    \draw (15) -- (46);
    \draw (15) -- (60);
    \draw (16) -- (17);
    \draw (16) -- (18);
    \draw (16) -- (23);
    \draw (16) -- (29);
    \draw (16) -- (31);
    \draw (16) -- (36);
    \draw (16) -- (38);
    \draw (16) -- (39);
    \draw (16) -- (54);
    \draw (16) -- (56);
    \draw (16) -- (57);
    \draw (17) -- (57);
    \draw (19) -- (22);
    \draw (19) -- (50);
    \draw (20) -- (35);
    \draw (20) -- (38);
    \draw (20) -- (41);
    \draw (20) -- (43);
    \draw (20) -- (55);
    \draw (20) -- (58);
    \draw (20) -- (59);
    \draw (21) -- (35);
    \draw (21) -- (41);
    \draw (21) -- (45);
    \draw (21) -- (51);
    \draw (22) -- (28);
    \draw (22) -- (30);
    \draw (22) -- (32);
    \draw (22) -- (35);
    \draw (22) -- (38);
    \draw (22) -- (45);
    \draw (22) -- (50);
    \draw (22) -- (53);
    \draw (22) -- (62);
    \draw (23) -- (36);
    \draw (24) -- (38);
    \draw (24) -- (49);
    \draw (24) -- (57);
    \draw (26) -- (27);
    \draw (26) -- (60);
    \draw (27) -- (60);
    \draw (28) -- (32);
    \draw (28) -- (45);
    \draw (28) -- (47);
    \draw (29) -- (39);
    \draw (30) -- (35);
    \draw (30) -- (38);
    \draw (30) -- (50);
    \draw (30) -- (62);
    \draw (32) -- (45);
    \draw (32) -- (53);
    \draw (33) -- (34);
    \draw (34) -- (35);
    \draw (34) -- (38);
    \draw (34) -- (52);
    \draw (35) -- (38);
    \draw (35) -- (42);
    \draw (35) -- (43);
    \draw (35) -- (49);
    \draw (35) -- (52);
    \draw (35) -- (53);
    \draw (35) -- (55);
    \draw (35) -- (58);
    \draw (35) -- (59);
    \draw (35) -- (61);
    \draw (35) -- (62);
    \draw (36) -- (39);
    \draw (37) -- (45);
    \draw (38) -- (42);
    \draw (38) -- (43);
    \draw (38) -- (45);
    \draw (38) -- (49);
    \draw (38) -- (52);
    \draw (38) -- (55);
    \draw (38) -- (58);
    \draw (38) -- (60);
    \draw (38) -- (61);
    \draw (38) -- (62);
    \draw (39) -- (57);
    \draw (40) -- (49);
    \draw (42) -- (49);
    \draw (42) -- (52);
    \draw (43) -- (59);
    \draw (44) -- (48);
    \draw (44) -- (51);
    \draw (44) -- (55);
    \draw (45) -- (47);
    \draw (45) -- (51);
    \draw (45) -- (53);
    \draw (48) -- (51);
    \draw (48) -- (55);
    \draw (49) -- (52);
    \draw (49) -- (60);
    \draw (49) -- (61);
    \draw (49) -- (62);
    \draw (52) -- (61);
    \draw (52) -- (62);
    \draw (53) -- (62);
    \draw (55) -- (58);
    \draw (55) -- (59);
    \draw (56) -- (57);
    \draw (58) -- (59);
    \draw (61) -- (62);

    \end{pgfonlayer}

    \matrix (m) [
      right=5.75cm of 38,
      matrix of nodes,
      column sep=2pt,
      nodes={inner sep=1, minimum size=0}
    ] {
      \node[vertex, fill=gray!30, inner sep=.75, minimum size=15] {38}; \\
      \midrule
      \node {\tc}; \\
      \node {\bc}; \\
      \node {\cc}; \\
      \node {\dc}; \\
      \node {\ev}; \\
      \node {\pr}; \\
    };
  \end{tikzpicture}
  \caption{\label{fig:911_hijackers}Krebs' 9/11 hijackers network.}
\end{figure}

The heatmaps in Figure~\ref{fig:ranksimilarity} illustrate mutual agreement
among the centrality measures. Figure~\ref{fig:binary_heatmap} depicts a
$20\times 5$ binary heatmap for each of the six centrality measures in this
study, where the rows are the test graphs in Table~\ref{tbl:graphdata} and the
columns are competitor measures. A non-empty $(i,j)$ entry in the binary heatmap
for a centrality measure indicates agreement with some other centrality measure
$j$ on the most central vertex in graph $i$. The overall pairwise similarity is
illustrated by the heatmap in Figure~\ref{fig:similarity}, which holds the
column sums from the binary heatmaps where the maximum value for a sum is the
number of graphs (20 in this case).

\begin{figure}[H]
  \centering
  \begin{subfigure}[t]{.6\linewidth}
    \centering
    \begin{tikzpicture}[
      nodes={
        font=\scriptsize
      },
      dotmatrix/.style={
        rectangle, draw, inner sep=0pt,
        matrix of nodes,
        nodes in empty cells,
        row sep=-\pgflinewidth,
        column sep=-\pgflinewidth,
        nodes={
          inner sep=0pt,
          outer sep=0pt,
          minimum width=8,
          minimum height=8,
          anchor=center,
        }
      }
    ]

    \matrix(tc) [%
      dotmatrix
    ]{
      {} & {} & \mdot & \mdot & {} \\
      {} & {} & {} & {} & {} \\
      {} & {} & \mdot & \mdot & \mdot \\
      \mdot & \mdot & \mdot & \mdot & \mdot \\
      {} & {} & {} & \mdot & {} \\
      {} & {} & {} & \mdot & {} \\
      \mdot & \mdot & \mdot & \mdot & \mdot \\
      {} & {} & {} & {} & {} \\
      {} & {} & {} & \mdot & {} \\
      {} & {} & {} & {} & {} \\
      {} & {} & {} & {} & {} \\ 
      {} & {} & {} & \mdot & {} \\
      {} & {} & {} & \mdot & {} \\
      {} & {} & \mdot & \mdot & {} \\
      {} & \mdot & {} & \mdot & {} \\
      {} & \mdot & \mdot & \mdot & {} \\
      \mdot & \mdot & \mdot & \mdot & \mdot \\
      {} & {} & {} & {} & {} \\
      {} & {} & {} & {} & {} \\
      {} & {} & \mdot & \mdot & {} \\
    };
    \foreach \i in {1, ..., 20} {
      \node [label={[label distance=0pt,font=\scriptsize]left:\i}] at (tc-\i-1) {};
    }
    \foreach \x [count=\i] in {\bc, \cc, \dc, \ev, \pr} {
      \node [
        pin={[pin distance=1ex]above:{}},
        label={[label distance=0pt,
            font=\tiny,
            xshift=-2pt,
            anchor=south west,
            rotate=45]above:\x}
      ] at (tc-1-\i) {};
    }

    \matrix(bc) [right=5pt of tc] [%
      dotmatrix
    ]{
      {} & {} & {} & {} & {} \\
      {} & \mdot & {} & {} & {} \\
      {} & \mdot & {} & {} & {} \\
      \mdot & \mdot & \mdot & \mdot & \mdot \\
      {} & \mdot & \mdot & {} & \mdot \\
      {} & \mdot & {} & {} & {} \\
      \mdot & \mdot & \mdot & \mdot & \mdot \\
      {} & {} & {} & {} & {} \\
      {} & \mdot & \mdot & {} & \mdot \\
      {} & {} & {} & {} & \mdot \\
      {} & \mdot & {} & {} & \mdot \\ 
      {} & {} & {} & {} & {} \\
      {} & {} & {} & {} & {} \\
      {} & {} & {} & {} & \mdot \\
      {} & {} & \mdot & {} & \mdot \\
      {} & {} & {} & {} & {} \\
      \mdot & \mdot & \mdot & \mdot & \mdot \\
      {} & {} & {} & {} & \mdot \\
      {} & {} & {} & {} & {} \\
      {} & {} & {} & {} & {} \\
    };
    \foreach \x [count=\i] in {\tc, \cc, \dc, \ev, \pr} {
      \node [
        pin={[pin distance=1ex]above:{}},
        label={[label distance=0pt,
            font=\tiny,
            xshift=-2pt,
            anchor=south west,
            rotate=45]above:\x}
      ] at (bc-1-\i) {};
    }

    \matrix(cc) [right=5pt of bc] [%
      dotmatrix
    ]{
      {} & {} & {} & {} & {} \\
      {} & \mdot & {} & {} & {} \\
      {} & \mdot & {} & {} & {} \\
      \mdot & \mdot & \mdot & \mdot & \mdot \\
      {} & \mdot & \mdot & {} & \mdot \\
      {} & \mdot & {} & {} & {} \\
      \mdot & \mdot & \mdot & \mdot & \mdot \\
      {} & {} & {} & {} & {} \\
      {} & \mdot & \mdot & {} & \mdot \\
      {} & {} & {} & {} & {} \\
      {} & \mdot & {} & {} & \mdot \\ 
      {} & {} & {} & {} & {} \\
      {} & {} & {} & {} & {} \\
      {} & {} & {} & {} & {} \\
      \mdot & {} & {} & \mdot & {} \\
      \mdot & {} & \mdot & \mdot & {} \\
      \mdot & \mdot & \mdot & \mdot & \mdot \\
      {} & {} & {} & {} & {} \\
      {} & {} & {} & {} & {} \\
      {} & {} & {} & {} & {} \\
    };
    \foreach \x [count=\i] in {\tc, \bc, \dc, \ev, \pr} {
      \node [
        pin={[pin distance=1ex]above:{}},
        label={[label distance=0pt,
            font=\tiny,
            xshift=-2pt,
            anchor=south west,
            rotate=45]above:\x}
      ] at (cc-1-\i) {};
    }

    \matrix(dc) [right=5pt of cc] [%
      dotmatrix
    ]{
      \mdot & {} & {} & \mdot & \mdot \\
      {} & {} & {} & \mdot & \mdot \\
      \mdot & {} & {} & \mdot & \mdot \\
      \mdot & \mdot & \mdot & \mdot & \mdot \\
      {} & \mdot & \mdot & {} & \mdot \\
      {} & {} & {} & {} & \mdot \\
      \mdot & \mdot & \mdot & \mdot & \mdot \\
      {} & {} & {} & \mdot & \mdot \\
      {} & \mdot & \mdot & {} & \mdot \\
      {} & {} & {} & \mdot & {} \\
      {} & {} & {} & \mdot & {} \\ 
      {} & {} & {} & {} & {} \\
      {} & {} & {} & {} & \mdot \\
      \mdot & {} & {} & \mdot & {} \\
      {} & \mdot & {} & {} & \mdot \\
      \mdot & {} & \mdot & \mdot & {} \\
      \mdot & \mdot & \mdot & \mdot & \mdot \\
      {} & {} & {} & {} & {} \\
      {} & {} & {} & {} & \mdot \\
      \mdot & {} & {} & \mdot & {} \\
    };
    \foreach \x [count=\i] in {\tc, \bc, \cc, \ev, \pr} {
      \node [
        pin={[pin distance=1ex]above:{}},
        label={[label distance=0pt,
            font=\tiny,
            xshift=-2pt,
            anchor=south west,
            rotate=45]above:\x}
      ] at (dc-1-\i) {};
    }
   
    \matrix(ev) [right=5pt of dc] [%
      dotmatrix
    ]{
      \mdot & {} & {} & \mdot & {} \\
      {} & {} & {} & \mdot & \mdot \\
      \mdot & {} & {} & \mdot & \mdot \\
      \mdot & \mdot & \mdot & \mdot & \mdot \\
      \mdot & {} & {} & {} & {} \\
      \mdot & {} & {} & {} & {} \\
      \mdot & \mdot & \mdot & \mdot & \mdot \\
      {} & {} & {} & \mdot & {} \\
      \mdot & {} & {} & {} & {} \\
      {} & {} & {} & \mdot & {} \\
      {} & {} & {} & \mdot & {} \\ 
      \mdot & {} & {} & {} & {} \\
      \mdot & {} & {} & {} & {} \\
      \mdot & {} & {} & \mdot & {} \\
      \mdot & {} & \mdot & {} & {} \\
      \mdot & {} & \mdot & \mdot & {} \\
      \mdot & \mdot & \mdot & \mdot & \mdot \\
      {} & {} & {} & {} & {} \\
      {} & {} & {} & {} & {} \\
      \mdot & {} & {} & \mdot & {} \\
    };
    \foreach \x [count=\i] in {\tc, \bc, \cc, \dc, \pr} {
      \node [
        pin={[pin distance=1ex]above:{}},
        label={[label distance=0pt,
            font=\tiny,
            xshift=-2pt,
            anchor=south west,
            rotate=45]above:\x}
      ] at (ev-1-\i) {};
    }
    
    \matrix(pr) [right=5pt of ev] [%
    dotmatrix
    ]{
      {} & {} & {} & \mdot & {} \\
      {} & {} & {} & \mdot & \mdot \\
      \mdot & {} & {} & \mdot & \mdot \\
      \mdot & \mdot & \mdot & \mdot & \mdot \\
      {} & \mdot & \mdot & \mdot & {} \\
      {} & {} & {} & \mdot & {} \\
      \mdot & \mdot & \mdot & \mdot & \mdot \\
      {} & {} & {} & \mdot & {} \\
      {} & \mdot & \mdot & \mdot & {} \\
      {} & \mdot & {} & {} & {} \\
      {} & \mdot & \mdot & {} & {} \\ 
      {} & {} & {} & {} & {} \\
      {} & {} & {} & \mdot & {} \\
      {} & \mdot & {} & {} & {} \\
      {} & \mdot & {} & \mdot & {} \\
      {} & {} & {} & {} & {} \\
      \mdot & \mdot & \mdot & \mdot & \mdot \\
      {} & \mdot & {} & {} & {} \\
      {} & {} & {} & \mdot & {} \\
      {} & {} & {} & {} & {} \\
    };
    \foreach \x [count=\i] in {\tc, \bc, \cc, \dc, \ev} {
      \node [
        pin={[pin distance=1ex]above:{}},
        label={[label distance=0pt,
            font=\tiny,
            xshift=-2pt,
            anchor=south west,
            rotate=45]above:\x}
      ] at (pr-1-\i) {};
    }

    \node [below=2pt of tc] {\tc};
    \node [below=2pt of bc] {\bc};
    \node [below=2pt of cc] {\cc};
    \node [below=2pt of dc] {\dc};
    \node [below=2pt of ev] {\ev};
    \node [below=2pt of pr] {\pr};
  \end{tikzpicture}
  \subcaption{\label{fig:binary_heatmap}}
  \end{subfigure}%
  \begin{subfigure}[t]{0.40\linewidth}
    \centering
    \sisetup{round-mode=places, round-precision=2}
    \begin{tikzpicture}[
        every node/.style = {
          outer sep=0,
          minimum size=20pt,
          align=center,
          font=\footnotesize
        },
        scale=0.70
      ]
      \foreach \x [count=\i] in {
        {20,3,5,8,14,4},  
        {3,20,9,6,3,10},  
        {5,9,20,6,5,6},   
        {8,6,6,20,12,13}, 
        {14,3,5,12,20,5}, 
        {4,10,6,13,5,20}, 
      }{
        \foreach \y [count=\j,
          evaluate=\y as \h using \y*100/20.0,
          evaluate=\y as \l using \y/20.0] in \x {
          \ifthenelse{\y=20}{
            \node[fill=black!\h!white, text=white] at (\j, -\i) {\y};
          }{
            \node[fill=black!\h!white, text=black] at (\j, -\i) {\y};
          }
        }
      }

      %
      \def\array{\tc,\bc,\cc,\dc,\ev,\pr}
      \foreach \x [count=\i, evaluate=\x as \k using \i+1] in \array {
        \xdef\n{\k}
      }

      \foreach \c [count=\i] in \array  {
        \node at (0, -\i) {\c};

        \node at (\i, -\n) {\c};
      }

      \node[draw=none] (A) at (1,-1) {};
      \node[draw=none] (D) at (6,-6) {};
      \draw (A.north west) rectangle (D.south east);

    \end{tikzpicture}
    \subcaption{\label{fig:similarity}}
  \end{subfigure}
  \caption{\label{fig:ranksimilarity}
    \subref{fig:binary_heatmap} Binary heatmap ($20\times 5$) for each
    centrality measure (labeled below); rows correspond to graphs numbered 1--20
    from Table~\ref{tbl:graphdata}, and columns are labeled by competitor
    centrality measures. \subref{fig:similarity} Heatmap of pairwise similarity
    with respect to ranking agreement.}

\end{figure}

Figure~\ref{fig:normality} summarizes how well each centrality measure aligns
with the norm as opposed to being an outlier. The overall percent agreement is
given in Figure~\ref{fig:agreement}, taken simply as the row (or column) sum in
the heatmap of Figure~\ref{fig:similarity}, ignoring the diagonal entry
(self-similarity). This is equivalent to the number of non-zero entries in a
binary heatmap. Figure~\ref{fig:commonality} depicts the degree of concordance,
meaning the percentage of graphs in which a centrality measure is in agreement
with at least one other. It is taken as the number of non-empty rows in the
binary heatmap. In essence, this captures the normality of a centrality
measure. A measure with a very low percent concordance indicates it is abnormal
with respect to identifying importance, which could mean it is too highly
specialized or esoteric. On the other hand, high concordance may indicate a lack
of novelty.

\begin{figure}[H]
  \centering
  \begin{subfigure}[t]{0.5\linewidth}
    \centering
    \begin{tikzpicture}[
        every node/.style = {
          outer sep=0,
          minimum size=20pt,
          align=center,
        },
      ]
      \begin{axis}[
          scale=0.70,
          ybar,
          ymin=25,
          enlargelimits=0.15,
          ylabel={Percent agreement},
          ylabel near ticks,
          symbolic x coords={\tc, \bc, \cc, \dc, \ev, \pr},
          ymajorgrids=true,
          major grid style={dotted, black},
          xtick pos=lower,
          ytick pos=left,
          xtick=data,
          tick label style={font=\footnotesize},
          ytick={30,35,40,45}
        ]
        \addplot [draw=gray,fill=gray] coordinates {
          (\tc, 34) (\bc, 31) (\cc, 31) (\dc, 45) (\ev, 39) (\pr, 38)  
        };
      \end{axis}
    \end{tikzpicture}
    \subcaption{\label{fig:agreement}}
  \end{subfigure}%
  \begin{subfigure}[t]{0.5\linewidth}
    \centering
    \begin{tikzpicture}[
        every node/.style = {
          outer sep=0,
          minimum size=20pt,
          align=center,
        },
      ]
      \begin{axis}[
          scale=0.70,
          ybar,
          enlargelimits=0.15,
          ylabel={Percent concordance},
          ylabel near ticks,
          symbolic x coords={\tc, \bc, \cc, \dc, \ev, \pr},
          ymajorgrids=true,
          major grid style={dotted, black},
          xtick pos=lower,
          ytick pos=left,
          xtick=data,
          tick label style={font=\footnotesize},
          ytick={60,70,80,90}
        ]
        \addplot [draw=gray,fill=gray] coordinates {
          (\tc, 70) (\bc, 65) (\cc, 55) (\dc, 90) (\ev, 90) (\pr, 85)
        };
      \end{axis}
    \end{tikzpicture}
    \subcaption{\label{fig:commonality}}
  \end{subfigure}
  \caption{\label{fig:normality}
    \subref{fig:agreement} Percent ranking agreement (row or column sum in
    \subref{fig:similarity} ignoring self-similarity).}
\end{figure}

Observe that triangle centrality agrees with 34\% of the choices made for the
most central vertex. It also has agreement with at least one other centrality
measure 70\% of the time. Conversely, in 30\% of the graphs, triangle centrality
identified a central vertex that differed from the central vertices given by the
other measures. We believe this level of concordance strikes a good balance
between normality and novelty. We also see that eigenvector centrality agrees
more often with triangle centrality than the other measures, selecting the same
central vertex in 14 of the graphs (70\%). This happens to coincide with the
concordance of triangle centrality, but reader should note it is not generally
true.

Now observe that closeness centrality has nine empty rows in its binary heatmap
leading to 55\% concordance and has only 31\% agreement overall. This pattern is
followed closely by betweenness centrality. These results suggest that closeness
and betweenness centrality may be less versatile in measuring importance. It is
also evident that degree centrality is the most mutually associated measure,
indicating the ubiquitous role that degree plays in structural importance. But
if a centrality measure does not differ significantly from degree centrality,
then it does not offer any new insight.

In our comparative analysis of centrality, we used each centrality measure to
answer only the most basic question: \emph{who is the most important?} The
answer to this question is relative to the centrality measure, and hence, the
top-ranked vertex by a centrality measure is, by virtue of definition, the most
important. But the notion of ``most'' central under even just a single
centrality measure is not well defined because there can be ties for the top
rank; moreover, ranks can be real numbers, and hence, the separation between the
top and next rank may be arbitrarily close. Nonetheless, nominating the most
important node in a network, after handling ties and numeric precision, is a
convenient point of analysis. To check the robustness on agreement across the
centrality measures, we will compare the top $k$ vertices ranked by each measure
using the Jaccard index~\cite{bib:jaccard} (see
Appendix~\ref{sec:jaccard_similarity} for a brief review of the Jaccard index).

Figure~\ref{fig:jaccard_summary} illustrates the Jaccard similarity between the
six centrality measures using the top 10 ($k=10$) ranked vertices in each of the
graphs from Table~\ref{tbl:graphdata}. The binary heatmap for a centrality
measure in Figure~\ref{fig:jaccard_heatmap} indicates for each test graph, which
other measure is most similar to it. Ties are broken by choosing the competing
measure that is first to rank a vertex in its top 10 list higher than the other
competitors (see Appendix~\ref{sec:jaccard_similarity} for full details). Hence,
each row in the binary heatmap has at most one entry. It follows that the column
with the most entries indicates the measure that is most similar overall to that
centrality measure. The heatmap in Figure~\ref{fig:jaccard_similarity}
summarizes the overall Jaccard similarity for each pair of centrality measures;
specifically, each row tallies the column sums from the corresponding binary
heatmap in Figure~\ref{fig:jaccard_heatmap}. An interested reader can find in
Appendix~\ref{sec:jaccard_similarity} a tabulation
(Table~\ref{tbl:best_jaccard}) of the six centrality measures with their closest
competitor by Jaccard index and the full table
(Table~\ref{tbl:allpairs_jaccard}) of all-pairs Jaccard similarity.

\begin{figure}[H]
  \centering
  \begin{subfigure}[t]{.6\linewidth}
    \centering
  \begin{tikzpicture}[
      nodes={
        font=\scriptsize
      },
      dotmatrix/.style={
        rectangle, draw, inner sep=0pt,
        matrix of nodes,
        nodes in empty cells,
        row sep=-\pgflinewidth,
        column sep=-\pgflinewidth,
        nodes={
          inner sep=0pt,
          outer sep=0pt,
          minimum width=8,
          minimum height=8,
          anchor=center,
        }
      }
    ]

    \matrix(tc) [%
      dotmatrix,
    ]{
      {} & {} & {} & \mdot & {} \\
      {} & \mdot & {} & {} & {} \\
      {} & {} & {} & \mdot & {} \\
      {} & {} & {} & \mdot & {} \\
      {} & {} & {} & \mdot & {} \\
      {} & {} & \mdot & {} & {} \\
      {} & {} & {} & \mdot & {} \\
      {} & {} & \mdot & {} & {} \\
      {} & {} & {} & \mdot & {} \\
      {} & {} & {} & \mdot & {} \\
      {} & {} & \mdot & {} & {} \\ 
      {} & {} & {} & \mdot & {} \\
      {} & {} & {} & \mdot & {} \\
      {} & {} & {} & \mdot & {} \\
      {} & {} & {} & \mdot & {} \\
      {} & {} & {} & \mdot & {} \\
      {} & {} & {} & \mdot & {} \\
      {} & {} & {} & \mdot & {} \\
      {} & {} & \mdot & {} & {} \\
      {} & {} & {} & \mdot & {} \\
    };
    \foreach \i in {1, ..., 20} {
      \node [label={[label distance=0pt,font=\scriptsize]left:\i}] at (tc-\i-1) {};
    }
    \foreach \x [count=\i] in {\bc, \cc, \dc, \ev, \pr} {
      \node [
        pin={[pin distance=1ex]above:{}},
        label={[label distance=0pt,
            font=\tiny,
            xshift=-2pt,
            anchor=south west,
            rotate=45]above:\x}
      ] at (tc-1-\i) {};
    }

    \matrix(bc) [right=5pt of tc] [%
      dotmatrix
    ]{
      {} & \mdot & {} & {} & {} \\
      {} & \mdot & {} & {} & {} \\
      {} & \mdot & {} & {} & {} \\
      {} & {} & {} & {} & \mdot \\
      {} & {} & {} & {} & \mdot \\
      {} & \mdot & {} & {} & {} \\
      {} & {} & {} & {} & \mdot \\
      {} & \mdot & {} & {} & {} \\
      {} & {} & \mdot & {} & {} \\
      {} & {} & {} & {} & \mdot \\
      {} & \mdot & {} & {} & {} \\ 
      {} & \mdot & {} & {} & {} \\
      {} & {} & {} & {} & \mdot \\
      {} & {} & {} & {} & \mdot \\
      {} & {} & {} & {} & \mdot \\
      {} & {} & {} & {} & \mdot \\
      {} & {} & {} & {} & \mdot \\
      {} & {} & {} & {} & \mdot \\
      {} & {} & \mdot & {} & {} \\
      {} & {} & {} & {} & {} \\
    };
    \foreach \x [count=\i] in {\tc, \cc, \dc, \ev, \pr} {
      \node [
        pin={[pin distance=1ex]above:{}},
        label={[label distance=0pt,
            font=\tiny,
            xshift=-2pt,
            anchor=south west,
            rotate=45]above:\x}
      ] at (bc-1-\i) {};
    }

    \matrix(cc) [right=5pt of bc] [%
      dotmatrix
    ]{
      {} & \mdot & {} & {} & {} \\
      {} & \mdot & {} & {} & {} \\
      {} & \mdot & {} & {} & {} \\
      {} & {} & \mdot & {} & {} \\
      {} & {} & \mdot & {} & {} \\
      {} & \mdot & {} & {} & {} \\
      {} & {} & \mdot & {} & {} \\
      {} & \mdot & {} & {} & {} \\
      {} & {} & \mdot & {} & {} \\
      {} & {} & \mdot & {} & {} \\
      {} & \mdot & {} & {} & {} \\ 
      {} & \mdot & {} & {} & {} \\
      {} & \mdot & {} & {} & {} \\
      {} & {} & \mdot & {} & {} \\
      \mdot & {} & {} & {} & {} \\
      {} & {} & \mdot & {} & {} \\
      {} & \mdot & {} & {} & {} \\
      {} & \mdot & {} & {} & {} \\
      {} & \mdot & {} & {} & {} \\
      {} & {} & {} & {} & {} \\
    };
    \foreach \x [count=\i] in {\tc, \bc, \dc, \ev, \pr} {
      \node [
        pin={[pin distance=1ex]above:{}},
        label={[label distance=0pt,
            font=\tiny,
            xshift=-2pt,
            anchor=south west,
            rotate=45]above:\x}
      ] at (cc-1-\i) {};
    }

    \matrix(dc) [right=5pt of cc] [%
      dotmatrix
    ]{
      {} & {} & {} & \mdot & {} \\
      {} & {} & {} & \mdot & {} \\
      {} & {} & {} & {} & \mdot \\
      {} & {} & \mdot & {} & {} \\
      {} & {} & {} & {} & \mdot \\
      {} & {} & {} & {} & \mdot \\
      {} & {} & \mdot & {} & {} \\
      {} & {} & {} & {} & \mdot \\
      {} & {} & {} & {} & \mdot \\
      {} & {} & {} & {} & \mdot \\
      {} & {} & {} & {} & \mdot \\ 
      {} & {} & {} & {} & \mdot \\
      {} & {} & {} & {} & \mdot \\
      \mdot & {} & {} & {} & {} \\
      {} & {} & {} & {} & \mdot \\
      {} & {} & \mdot & {} & {} \\
      {} & {} & {} & \mdot & {} \\
      {} & \mdot & {} & {} & {} \\
      {} & {} & {} & {} & \mdot \\
      {} & {} & {} & {} & \mdot \\
    };
    \foreach \x [count=\i] in {\tc, \bc, \cc, \ev, \pr} {
      \node [
        pin={[pin distance=1ex]above:{}},
        label={[label distance=0pt,
            font=\tiny,
            xshift=-2pt,
            anchor=south west,
            rotate=45]above:\x}
      ] at (dc-1-\i) {};
    }

    \matrix(ev) [right=5pt of dc] [%
      dotmatrix
    ]{
      \mdot & {} & {} & {} & {} \\
      {} & {} & {} & \mdot & {} \\
      {} & {} & {} & \mdot & {} \\
      \mdot & {} & {} & {} & {} \\
      \mdot & {} & {} & {} & {} \\
      \mdot & {} & {} & {} & {} \\
      {} & {} & {} & \mdot & {} \\
      {} & {} & {} & \mdot & {} \\
      \mdot & {} & {} & {} & {} \\
      \mdot & {} & {} & {} & {} \\
      {} & {} & {} & \mdot & {} \\ 
      \mdot & {} & {} & {} & {} \\
      \mdot & {} & {} & {} & {} \\
      \mdot & {} & {} & {} & {} \\
      \mdot & {} & {} & {} & {} \\
      \mdot & {} & {} & {} & {} \\
      {} & {} & {} & \mdot & {} \\
      \mdot & {} & {} & {} & {} \\
      \mdot & {} & {} & {} & {} \\
      \mdot & {} & {} & {} & {} \\
    };
    \foreach \x [count=\i] in {\tc, \bc, \cc, \dc, \pr} {
      \node [
        pin={[pin distance=1ex]above:{}},
        label={[label distance=0pt,
            font=\tiny,
            xshift=-2pt,
            anchor=south west,
            rotate=45]above:\x}
      ] at (ev-1-\i) {};
    }

    \matrix(pr) [right=5pt of ev] [%
      dotmatrix
    ]{
      {} & {} & {} & \mdot & {} \\
      {} & {} & {} & \mdot & {} \\
      {} & {} & {} & \mdot & {} \\
      {} & \mdot & {} & {} & {} \\
      {} & \mdot & {} & {} & {} \\
      {} & {} & {} & \mdot & {} \\
      {} & \mdot & {} & {} & {} \\
      {} & {} & {} & \mdot & {} \\
      {} & {} & {} & \mdot & {} \\
      {} & {} & {} & \mdot & {} \\
      {} & \mdot & {} & {} & {} \\ 
      {} & {} & {} & \mdot & {} \\
      {} & {} & {} & \mdot & {} \\
      {} & {} & \mdot & {} & {} \\
      {} & {} & {} & \mdot & {} \\
      {} & \mdot & {} & {} & {} \\
      {} & \mdot & {} & {} & {} \\
      {} & \mdot & {} & {} & {} \\
      {} & {} & {} & \mdot & {} \\
      {} & {} & {} & \mdot & {} \\
    };
    \foreach \x [count=\i] in {\tc, \bc, \cc, \dc, \ev} {
      \node [
        pin={[pin distance=1ex]above:{}},
        label={[label distance=0pt,
            font=\tiny,
            xshift=-2pt,
            anchor=south west,
            rotate=45]above:\x}
      ] at (pr-1-\i) {};
    }

    \node [below=2pt of tc] {\tc};
    \node [below=2pt of bc] {\bc};
    \node [below=2pt of cc] {\cc};
    \node [below=2pt of dc] {\dc};
    \node [below=2pt of ev] {\ev};
    \node [below=2pt of pr] {\pr};
  \end{tikzpicture}
  \subcaption{\label{fig:jaccard_heatmap}}
  \end{subfigure}
  \centering
  \begin{subfigure}[t]{.4\linewidth}
    \centering
    \begin{tikzpicture}[
        every node/.style = {
          outer sep=0,
          minimum size=20pt,
          align=center,
          font=\footnotesize
        },
        scale=0.70
      ]
      \foreach \x [count=\i] in {
        {20,0,1,4,15,0},  
        {0,20,7,2,0,10},  
        {1,11,20,7,0,0},  
        {1,1,3,20,3,12},  
        {14,0,0,6,20,0},  
        {0,7,1,12,0,20},  
      }{
        \foreach \y [count=\j,
          evaluate=\y as \h using \y*100/20.0,
          evaluate=\y as \l using \y/20.0] in \x {
          \ifthenelse{\y=20}{
            \node[fill=black!\h!white, text=white] at (\j, -\i) {\y};
          }{
            \node[fill=black!\h!white, text=black] at (\j, -\i) {\y};
          }
        }
      }

      %
      \def\array{\tc,\bc,\cc,\dc,\ev,\pr}
      \foreach \x [count=\i, evaluate=\x as \k using \i+1] in \array {
        \xdef\n{\k}
      }

      \foreach \c [count=\i] in \array  {
        \node[draw=none] at (0, -\i) {\c};

        \node[draw=none] at (\i, -\n) {\c};
      }

      \node[draw=none] (A) at (1,-1) {};
      \node[draw=none] (D) at (6,-6) {};
      \draw (A.north west) rectangle (D.south east);
      
    \end{tikzpicture}
    \subcaption{\label{fig:jaccard_similarity}}
  \end{subfigure}
\caption{\label{fig:jaccard_summary}
  \subref{fig:jaccard_heatmap} Jaccard-based binary heatmap ($20\times 5$) for
  each centrality measure (labeled below); rows correspond to graphs numbered
  1--20 from Table~\ref{tbl:graphdata}, and columns are labeled by competitor
  centrality measure. \subref{fig:jaccard_similarity} Heatmap of overall Jaccard
  similarities between centrality measures (column sums
  from~\subref{fig:jaccard_heatmap}).}

\end{figure}

For each centrality measure in Figure~\ref{fig:jaccard_similarity}, the most
similar measure to it (row-wise) coincides with the most similar in
Figure~\ref{fig:similarity}. This suggests that the top-ranked vertex was a
reasonable proxy in comparing measures. But observe that the overall Jaccard
similarity relationship is not symmetric between the measures. The results in
Figure~\ref{fig:jaccard_similarity} demonstrate that a measure $j$ can be the
most similar to $k$ but $k$ may not be the most similar to $j$. This is evident
between betweenness and closeness centralities. It is clear in
Figure~\ref{fig:jaccard_summary} that each centrality has a competitor measure
that is predominantly more similar to it than the others. The degree of
similarity is indicated by the count given in
Figure~\ref{fig:jaccard_similarity}. From this data, we again surmise that
eigenvector and triangle centrality more often agree on the most central
vertices than with the other measures. Our results show that closeness is most
similar to betweenness, betweenness is most similar to PageRank, PageRank is
most similar to degree, and degree is most similar to PageRank. These results
also support our earlier observation that degree centrality is often similar to
the other measures as indicated by the degree centrality column in
Figure~\ref{fig:jaccard_similarity}.

Overall, we conclude that triangle centrality finds the central vertex in many
of the same graphs as other measures and aligns with the consensus when there is
no ambiguity. Yet, it uniquely identified central vertices in 30\% of the
graphs, which suggests it is not overly specialized. It is therefore a versatile
measure that is able to center on characteristics that are missed by other
measures. This makes it a valuable and complementary tool for graph centrality
analysis. Moreover, it is asymptotically faster to compute on sparse graphs than
the other measures with the exception of degree centrality.

\section{\label{sec:algorithm}Algorithm}
The first task in computing triangle centrality is to get the triangle count of
every vertex in the graph. This task can be achieved by any efficient triangle
counting algorithm in $O(m\bar\delta)$ time. See~\cite{bib:ortmann_brandes2014,
  bib:latapy2008} for a review. The second task is summing the triangle counts
from neighbors of a vertex, with separate sums for triangle and non-triangle
neighbors. This requires first identifying the specific triangle neighbors of
each vertex. Then, for every vertex, we need to calculate the core triangle sum,
$\sum_{u\in N_\triangle[v]} \triangle(u)$, and the non-core triangle sum,
$\sum_{w\in \{N(v) \setminus N_\triangle(v)\}} \triangle(w)$. Given these sums
and $\triangle(G)$, the triangle centrality follows. This procedure leads to our
elementary algorithm in Figure~\ref{fig:tricent_elementary}. We will refer to
the elementary steps in Figure~\ref{fig:tricent_elementary} throughout the
remaining algorithm descriptions and results.

\begin{figure}[H]
\centering
\fbox{%
\begin{minipage}{.90\textwidth}
\begin{enumerate}
\item\label{itm:tricent_step1}
  For each vertex $v\in V$, compute the local triangle count, $\triangle(v)$,
  and update $\triangle(G)$
\item\label{itm:tricent_step2}
  For each vertex $v\in V$, find and store $v$'s triangle neighbors
\item\label{itm:tricent_step3}
  For each vertex $v\in V$ do
  \begin{enumerate}[label = \roman*., ref=\theenumi.\roman*]
  \item\label{itm:tricent_step3_i}
    Calculate the core triangle sum, $x = \sum_{u\in
      N_\triangle[v]} \triangle(u)$
  \item \label{itm:tricent_step3_ii}
    Sum all neighbor triangle counts, $s = \sum_{u\in N(v)} \triangle(u)$
  \item \label{itm:tricent_step3_iii}
    Get non-core triangle sum,
    $\sum_{w\in \{N(v)\setminus N_\triangle(v)\}} \triangle(w) =
    s - x + \triangle(v)$
  \item \label{itm:tricent_step3_iv}
    Compute the triangle centrality $\tc(v) =
    \frac{\frac{1}{3}\sum_{u\in N_\triangle[v]} \triangle(u)
      + \sum_{w\in \{N(v)\setminus N_\triangle(v)\}}
      \triangle(w)}{\triangle(G)}$
  \end{enumerate}
\end{enumerate}
\end{minipage}
}
\caption{\label{fig:tricent_elementary}Triangle centrality elementary
  algorithm.}
\end{figure}

Since a triangle is symmetric, it can be reported by the lowest degree vertex in
the triangle. Hence, with degree-ordering, all triangles can be counted and
listed in optimal time. This has been well known since 1985 when Chiba and
Nishizeki~\cite{bib:chiba_nishizeki1985} listed triangles in $m\bar\delta \le
2ma = O(m\sqrt{m})$ time. Our \textproc{triangleneighbor} algorithm given in
Procedure~\ref{alg:trineighbor} applies the ideas
from~\cite{bib:chiba_nishizeki1985} to efficiently get triangle counts and list
triangle neighbors.

\setcounter{algorithm}{0}
\begin{algorithm}[H]
  \floatname{algorithm}{Procedure}
  \caption{\label{alg:trineighbor}TriangleNeighbor}
  \begin{algorithmic}[1]
    \Require $L$, array to hold triangle neighbor lists for each $v\in V$.
    \Require $T$, zero-initialized array of size $n$.
    \For{$v\in V$}
      \For{$u\in N_H(v)$}
        \State set $T(u) := 1$
      \EndFor
      \For{$u\in N_L(v))$}
        \label{ln:trineighbor_start}
        \For{$w\in N_H(u)$}
          \If{$T(u)$ is $1$}
            \State increment
             $\triangle(v),\triangle(u),\triangle(w),\triangle(G)$
            \State $L(v) \gets u, L(v) \gets w$
            \State $L(u) \gets v, L(v) \gets w$
            \State $L(w) \gets v, L(v) \gets u$
          \EndIf
        \EndFor
      \EndFor
      \label{ln:trineighbor_end}  
      \For{$u\in N_H(v)$}
        \State set $T(u) := 0$
      \EndFor
    \EndFor
  \end{algorithmic}
\end{algorithm}

First observe that degree-ordering prevents directed cycles so that each $\{u,
v, w\}$ triangle is a directed acyclic subgraph, like the one depicted in
Figure~\ref{fig:ordered_triangle} corresponding to the ordering $\pi(u) < \pi(v)
< \pi(w)$. Observe that Procedure~\ref{alg:trineighbor} detects each triangle
once from the $(u,v)$ oriented edge. This leads to the following result.

\begin{figure}[H]
\centering
\begin{tikzpicture}
  [scale=.75, inner sep=1.5, minimum size=14, node distance=.5,
  vertex/.style={circle,draw=black,thin},
  edge/.style={thin}]
  \node [vertex] (u) {u};
  \node [vertex] (v) [above right=of u] {v};
  \node [vertex] (w) [below right=of u] {w};
  \draw [edge,->,>=latex] (u) to (v);
  \draw [edge,->,>=latex] (u) to (w);
  \draw [edge,->,>=latex] (v) to (w);
\end{tikzpicture}
\caption{\label{fig:ordered_triangle}Degree-ordered triangle pattern.}
\end{figure}

\begin{lemma}
  \label{lem:trineighbor}
  It is possible to count triangles and find triangle neighbors in
  $O(m\bar\delta)$ time and $O(m+n)$ space for all vertices in $G$.
\end{lemma}

\begin{proof}
  Procedure~\ref{alg:trineighbor} achieves the claim as follows.

  Each adjacency set $N(v)$ can be partitioned into subsets $N_L(v),N_H(v)$ in
  linear time without additional space by maintaining a back-pointer and
  swapping elements. The operations in
  lines~\ref{ln:trineighbor_start}--\ref{ln:trineighbor_end} are carried out
  only for the neighborhood set of the lower-degree vertex in each edge, taking
  $\sum_{(v,u)\in E} \min\{d(v),d(u)\}$ time and hence $O(m\bar\delta)$
  time. The remaining loops take $\sum_{v\in V} d(v) = O(m)$ time. All other
  operations take $O(1)$ time. Therefore, in total, it takes $O(m\bar\delta)$
  time and $O(m+n)$ space.
\end{proof}

We use Procedure~\ref{alg:trineighbor} to find triangle neighbors
(Step~\ref{itm:tricent_step2} of Figure~\ref{fig:tricent_elementary}) in our
main triangle centrality algorithm described next. An alternative triangle
neighbor algorithm using a single pass that may be more amenable to parallel
implementations is described in Appendix~\ref{sec:trineighbor_alt}.

\subsection{\label{sec:main_alg}Main Algorithm}
Triangle centrality requires the partitioned triangle counts between triangle
and non-triangle neighbors, which can be determined in $O(m\bar\delta)$ time
according to Lemma~\ref{lem:trineighbor}. We can now assert
Theorem~\ref{thm:tricent} to establish the complexity of computing triangle
centrality. Our Algorithm~\ref{alg:tricent_main} serves as a basis for an
efficient implementation.

\setcounter{algorithm}{1}
\begin{algorithm}[H]
\caption{\label{alg:tricent_main}}
  \begin{algorithmic}[1]
    \Require $X$, array of size $n$ indexed by $v\in V$, initialized to zero
    \Require $L$, array to hold triangle neighbor lists for each $v\in V$.
    \State Call \Call{TriangleNeighbor}{}
    \Comment Steps~\ref{itm:tricent_step1},~\ref{itm:tricent_step2}
    \For{$v\in V$}
      \Comment Step~\ref{itm:tricent_step3}
      \For{$u\in L(v)$}
        \State set $X(v) := X(v) + \triangle(u)$
      \EndFor
      \State set $x := X(v) + \triangle(v)$
      \Comment core triangle sum, Step~\ref{itm:tricent_step3_i}
      \For{$u\in N(v)$}
        \State set $s := s + \triangle(u)$
        \Comment neighbor triangle sum, Step~\ref{itm:tricent_step3_ii}
      \EndFor
      \State set $y := s - x + \triangle(v)$
      \Comment non-core triangle sum, Step~\ref{itm:tricent_step3_iii}
      \State output $\tc(v) = \bigl(\frac{1}{3}x + y\bigr)/\triangle(G)$
      \Comment Definition~\ref{def:tricent}, Step~\ref{itm:tricent_step3_iv}
    \EndFor
  \end{algorithmic}
\end{algorithm}

\begin{theorem}
  \label{thm:tricent}
  Triangle centrality can be computed in $O(m\bar\delta)$ time and $O(m+n)$ space
  for all vertices in $G$.
\end{theorem}

\begin{proof}
  Algorithm~\ref{alg:tricent_main} achieves this. The triangle counts and
  triangle neighbors are computed by \textproc{triangleneighbor}. It follows
  from Lemma~\ref{lem:trineighbor} that this takes $O(m\bar\delta)$ time.

  Summing triangle counts from triangle and non-triangle neighbors each takes
  $\sum_{v\in V} d(v) = 2m$ time.

  Each $L_v$ triangle neighbor list takes $d(v)$ space, thus overall these take
  $O(m)$ space. An additional $O(n)$ space is needed to hold triangle counts and
  the triangle core sum for every vertex.
\end{proof}

Next, we describe a linear algebraic algorithm for triangle centrality. For
sparse graphs, this algebraic algorithm matches the bounds of the combinatorial
algorithms introduced earlier. But the algebraic algorithm may have practical
benefits because it can leverage highly optimized matrix libraries. In recent
years, there has been interest in applying decades of experience with matrix
computation in the design of linear algebraic graph
algorithms~\cite{bib:kepner_gilbert2011, bib:buluc_madduri2011, bib:bucker2014,
  bib:yang2015, bib:azad_buluc2017, bib:yang2022, bib:burkhardt2021},
culminating into the GraphBLAS specification and implementations that can
already compute the graph triangle matrix~\cite{bib:davis2018, bib:davis2019,
  bib:suitesparsegraphblas, bib:buluc_graphblas2017, bib:yang2018}.

\subsection{\label{sec:algebraic_alg}Algebraic Algorithm}
Given the graph triangle matrix $\mat{T}$, the linear algebraic triangle
centrality defined in Proposition~\ref{prop:tricent_algebraic} is simple to
compute, requiring just two matrix-vector products, two matrix additions, and an
inner-product. Our Algorithm~\ref{alg:tricent_algebraic} computes it in optimal
time and linear space for sparse graphs if $\mat{T}$ can be constructed in
$O(m\bar\delta)$ time and $O(m+n)$ space.

\begin{algorithm}[H]
  \caption{\label{alg:tricent_algebraic}}
  \begin{algorithmic}[1]
    \Require $\mat{A}$, adjacency matrix in sparse matrix representation
    \Require $\mat{T}=\mat{A}^2\circ \mat{A}$, graph triangle matrix in sparse
    matrix representation
    \State Create binary matrix $\binr{\mat{T}}$ 
    \State set $\mat{X} := 3\mat{A} - 2\binr{\mat{T}} + \mat{I}$
    \State set $\vec{y} := \mat{T}\vec{1}$
    \State set $k := \vec{1}^\transpose \vec{y}$
    \State output $\cntvec{\tc} := \frac{1}{k}\mat{X}\vec{y}$
    \Comment Proposition~\ref{prop:tricent_algebraic}
  \end{algorithmic}
\end{algorithm}

\begin{claim}
  \label{clm:triangle_matrix}
  A sparse matrix representing $\mat{T}=\mat{A}^2\circ \mat{A}$ can be built in
  $O(m\bar\delta)$ time and $O(m+n)$ space.
\end{claim}

\begin{proof}
  First, count triangles and for each unique $(u,v)$ triangle edge, update the
  list of triangle neighbors for $u$ with $v$ and vice versa. This can be
  accomplished in $O(m\bar\delta)$ time using \textproc{triangleneighbor}
  according to Lemma~\ref{lem:trineighbor}. At the end of counting, iterate over
  the triangle neighbor lists and output weighted edges
  $(v,u,\triangle(u))$. Since each $v$ has at most $d(v)$ triangle neighbors,
  then writing the edges takes $\sum_{v\in V} d(v) = O(m)$ time. Finally, build
  a sparse matrix $\mat{T}$ on these weighted edges, which takes $O(m)$
  time. This $\mat{T}$ is equivalent to the graph triangle matrix produced by
  $\mat{A}^2\circ \mat{A}$. Storing the triangle counts to complete the triangle
  neighbor edge list takes $O(n)$ space. All triangle neighbors and hence
  non-zeros in $\mat{T}$ take $O(m)$ space. Thus, in total, $\mat{T}$ can be
  built in $O(m\bar\delta)$ time and $O(m+n)$ space.
\end{proof}

\begin{theorem}
  \label{thm:tricent_algebraic}
  There is a linear algebraic algorithm that computes the triangle centrality in
  $O(m\bar\delta)$ time and $O(m+n)$ space for all vertices in $G$ given a
  sparse matrix $\mat{A}$.
\end{theorem}

\begin{proof}
  Algorithm~\ref{alg:tricent_algebraic} achieves this, computing triangle
  centrality in accordance with Proposition~\ref{prop:tricent_algebraic} using
  sparse matrix representation for both the adjacency matrix $\mat{A}$ and graph
  triangle matrix $\mat{T}$. Thus, all operations on these matrices are over
  non-zeros only.

  It was established by Claim~\ref{clm:triangle_matrix} that the matrix
  $\mat{T}$ can be built in $O(m\bar\delta)$ time and $O(m+n)$ space. Now recall
  that both $\mat{A}$ and $\mat{T}$ have $O(m)$ non-zero values. Then, setting
  all non-zero values to unity in $\mat{T}$ to construct $\binr{\mat{T}}$ takes
  $O(m)$ time. Scalar operations and matrix additions on these matrices take
  $O(m)$ time. Therefore, it takes $O(m)$ time to produce $3\mat{A}$ and
  $2\binr{\mat{T}}$, and subsequently, the matrix addition
  $3\mat{A}-2\binr{\mat{T}}+\mat{I}$ also takes $O(m)$ time.

  A sparse matrix-vector multiplication takes $O(m)$ time. There are two
  matrix-vector multiplications in the algorithm. The first produces the
  $\mat{T}\vec{1}$ vector, and the second is the product of the
  $3\mat{A}-2\binr{\mat{T}}+\mat{I}$ matrix with this $\mat{T}\vec{1}$
  vector. The inner product between $\vec{1}^\transpose$ and the
  $\mat{T}\vec{1}$ vector takes $O(n)$ time. Therefore, all algebraic
  multiplications take $O(m+n)$ total time.

  Since $\mat{T}$ holds only triangle neighbors, then $\mat{T}$ and
  $\binr{\mat{T}}$ take $O(m)$ space. The effective addition of triangle
  neighbors to non-triangle neighbors resulting from the
  $3\mat{A}-2\binr{\mat{T}}+\mat{I}$ matrix addition leads to the same amount of
  space as required by $\mat{A}$. Thus, these matrices combined take $O(m)$
  space. All other space is for holding the vector $\mat{T}\vec{1}$, which takes
  $O(n)$ space.

  Altogether with the construction of $\mat{T}$ and $\binr{\mat{T}}$, the
  algorithm completes in $O(m\bar\delta)$ time and $O(m+n)$ space.
\end{proof}

Using fast matrix multiplication, we can calculate triangle centrality in
$n^{\omega+o(1)}$ time or similar bounds in terms of $m$ using the result of
Yuster and Zwick~\cite{bib:yuster_zwick2005}.

\begin{theorem}
  \label{thm:tricent_fastmatrix}
  Triangle centrality can be computed in $n^{\omega+o(1)}$ time using fast
  matrix multiplication, where $\omega$ is the matrix multiplication exponent,
  for all vertices in $G$.
\end{theorem}

The proof follows as an immediate consequence of fast matrix products in
computing the triangle centrality formulation given in
Proposition~\ref{prop:tricent_algebraic}.

\section{\label{sec:parallel_alg}Parallel Algorithm}
\subsection{\label{sec:pram}PRAM Algorithm}
We will describe PRAM algorithms for triangle centrality in this section. In a
PRAM~\cite{bib:fortune1978}, each processor can access any global memory
location in unit time. Processors can read from global memory, perform a
computation, and write a result to global memory in a single clock cycle. All
processors execute these instructions at the same time. Writes to a memory
location are restricted to a one processor at a time in a CREW PRAM.

It follows directly from matrix multiplication that the algebraic triangle
centrality (Proposition~\ref{prop:tricent_algebraic}) can be computed on a CREW
PRAM in $O(\log n)$ time using $O(n^3)$ processors. The work is bounded by the
matrix multiplication of $\mat{A}^2$ needed to get the graph triangle matrix
$\mat{T}$, and it is well known that multiplying two $n\times n$ matrices takes
$O(\log n)$ time using $O(n^3)$ CREW processors~\cite{bib:jaja1992}.

However, we can show that triangle centrality can be computed in $O(\log n)$
time using $O(m\sqrt{m})$ CREW processors and is therefore work-optimal up to a
logarithmic factor.\footnote{Work for parallel processing is $T(n)\times p$ for
runtime $T(n)$ using $p$ processors. It is optimal if it equals the sequential
runtime.} This is achieved by Algorithm~\ref{alg:tricent_pram}, in which
statements contained within a \textbf{for all} construct are performed
concurrently and all statements are executed in top-down order.

\begin{algorithm}[H]
  \caption{\label{alg:tricent_pram}}
  \begin{algorithmic}[1]
    \Require Array $P_v$ of size $2d(v)$ for each $v\in V$ initialized to zero
    \Require Array $X_v$ of size $d(v)$ for each $v\in V$ initialized to zero
    \Require Arrays $T_v, S_v$ of size $d(v)$ for each $v\in V$ initialized to
    zero
    \Require Processor $p$ for each $\{u,w\}$ pair in $N_H(v)$ assigned to
    cell $p$ in $P_w, P_u, P_v, X_w, X_u, X_v$.
    \ForAll{$\{u,w\} \in N_H(v), \forall v\in V$}
      \If{$\{u,w\} \in E$ and $\pi(v) < \pi(u) < \pi(w)$}
        \State write $v,u$ to $P_w[p]$,
        and $v,w$ to $P_u[p]$,
        and $u,w$ to $P_v[p]$
        \State set $X_w[p] := 1$, $X_u[p] := 1$, $X_v[p] := 1$
      \EndIf
    \EndFor
    \ForAll{$v\in V$}
      \State parallel sum over $X_v$ and set sum to $\triangle(v)$
    \EndFor
    \ForAll{$v\in V$}
      \State parallel sum over all $\triangle(v)$ and set sum to $\triangle(G)$
      \Comment Step~\ref{itm:tricent_step1}
      \State parallel sort/scan over $P_v$, write to $T_v$
      \Comment Step~\ref{itm:tricent_step2}, $T_v$ stores $N_\triangle(v)$
    \EndFor
    \ForAll{$u \in T_v, \forall v\in V$}
      \State set $T_v[u] := \triangle(u)$
      \Comment replace each triangle neighbor $u$ of $v$ with $\triangle(u)$
    \EndFor
    \ForAll{$(u,v) \in E$}
      \State set $S_v[u] := \triangle(u)$
    \EndFor
    \ForAll{$v\in V$}
      \State parallel sum over $T_v$ and set sum to $x$
      \Comment Step~\ref{itm:tricent_step3_i}, sans $\triangle(v)$
      \State parallel sum over $S_v$ and set sum to $s$
      \Comment Step~\ref{itm:tricent_step3_ii}
    \EndFor
    \ForAll{$v\in V$}
      \State set $x := x + \triangle(v)$
      \State set $y := s - x + \triangle(v)$
      \Comment non-core triangle sum, Step~\ref{itm:tricent_step3_iii}
      \State output $\tc(v) = \bigl(\frac{1}{3}x + y\bigr)/\triangle(G)$
      \Comment Definition~\ref{def:tricent}, Step~\ref{itm:tricent_step3_iv}
    \EndFor
  \end{algorithmic}
\end{algorithm}

\begin{theorem}
  \label{thm:tricent_crew}
  Triangle centrality can be computed on a CREW PRAM in $O(\log n)$ time using
  $O(m\sqrt{m})$ processors for all vertices in $G$.
\end{theorem}

\begin{proof}
  Observe there are $O(m\sqrt{m})$ processors for the first step. Each of these
  processors concurrently reads its assigned $u,w$ pair from $N_H(v)$ and if
  $\{u,w\}\in E$ and $\pi(v) < \pi(u) < \pi(w)$, then a unique triangle is
  found. Subsequently, each unique triangle can only be found by a unique
  processor $p$. Each $p$ exclusively writes the three unique triangle vertex
  pairs to its designated cell in the arrays $P_w, P_u, P_v$. That same
  processor also exclusively writes 1 to it designated cell in the arrays $X_w,
  X_u, X_v$. This first step takes $O(1)$ time using $O(m\sqrt{m})$ processors.

  The remaining steps in the algorithm utilizes parallel sum and scan primitives
  which are known to take $O(\log n)$ time on a CREW~\cite{bib:jaja1992}.

  Then, the triangle counts are computed by parallel sum over each $X_v$. Since
  $X_v$ is $O(d(v))$ in size, then it takes $O(m)$ processors and $O(\log n)$
  time to compute and store the counts. A second parallel sum to get
  $\triangle(G)$ takes the same time and number of processors. Finally, getting
  the unique triangle neighbors of each vertex requires a parallel sort and scan
  over each $P_v$. Parallel sorting can be accomplished in $O(\log n)$
  time~\cite{bib:cole1988} and scanning to remove duplicates takes the same
  time. Thus, altogether using $O(m)$ processors, the triangle neighbors are
  written to each $T_v$ in $O(\log n)$ time using $O(m)$ processors.

  Then, each triangle neighbor $u$ in $T_v$ is replaced with
  $\triangle(u)$. Given a processor for each of the $d(v)$ cells in $T_v$, then
  for all $v\in V$ it takes $O(m)$ processors to exclusively overwrite the
  entries in $O(1)$ time. Similarly, given a processor for each $(u,v) \in E$
  edge that writes to the corresponding $S_v[u]$ cell, the $\triangle(u)$ for
  all $d(v)$ neighbors of $u$ of $v$ are exclusively written in each $S_v$. This
  takes $O(1)$ time and $\sum_v d(v) = O(m)$ processors.

  Next, the sum of triangle counts for neighbors are obtained by parallel sum
  over the arrays $T_v,S_v$, taking $O(\log n)$ time and $O(m)$
  processors. Finally, the triangle centrality for each $v\in V$ is computed in
  parallel, taking $O(1)$ time and $O(n)$ processors.

  Therefore, in total, the triangle centrality can be computed in $O(\log n)$
  time using $O(m\sqrt{m})$ CREW processors.
\end{proof}

\subsection{\label{sec:mapreduce}MapReduce Algorithm}
The MapReduce model~\cite{bib:mrc2010,bib:goodrich2011,bib:pietracaprina2012} is
used to design distributed computing algorithms where efficiency is
parameterized by the number of rounds and communication bits. The model appeared
some years after the programming paradigm was popularized by
Google~\cite{bib:mapreduce2004}. It has been successfully employed in practice
for massive-scale algorithms~\cite{bib:chennupati2020, bib:harvey2018,
  bib:burkhardt2017, bib:burkhardt2015, bib:kolda2014, bib:kang2011,
  bib:vernica2010, bib:tsourakakis2009}. Algorithms in MapReduce use map and
reduce functions, executed in sequence. The input is a set of $\langle key,value
\rangle$ pairs that are ``mapped'' by instances of the map function into a
multiset of $\langle key,value \rangle$ pairs. The map output pairs are
``reduced'' and also output as a multiset of $\langle key,value \rangle$ pairs
by instances of the reduce function where a single reduce instance gets all
values associated with a key. A \emph{round} of computation is a single sequence
of map and reduce executions where there can be many instances of map and reduce
functions. Each map or reduce function can complete in polynomial time for input
$n$. Each map or reduce instance is limited to $O(n^{1-\epsilon})$ memory for a
constant $\epsilon > 0$, and an algorithm is allowed $O(n^{2-2\epsilon})$ total
memory. The number of machines/processors is bounded to $O(n^{1-\epsilon})$, but
each machine can run more than one instance of a map or reduce function.

We give a straightforward, 4-round MapReduce algorithm for triangle centrality
in Algorithm~\ref{alg:tricent_mapreduce}. The basic procedure is to list all
triangle edges, then separately combine the endpoints of these edges with the
original edge endpoints, and finally accumulate the endpoint counts for each
vertex to get the triangle counts and compute the triangle centrality. We will
show that it takes $O(1)$-rounds and communicates $O(m\sqrt{m})$ bits. Our
approach does not require storing any previous state between rounds and is
simple to implement. The MapReduce rounds are described next. The input is
presumed to be degree-annotated edges in the form of $\langle
(v,d(v)),(u,d(u))\rangle$ key-value pairs.

In the first round, the map function returns only the degree-ordered edges. This
ensures that the subsequent reduce function communicates $O(\sqrt{m}d(v))$
unique neighbor pairs for each vertex. Thus, the reduce function in this round
gets $\langle v, \{u \mid u\in N(v)\} \rangle$ and returns $\langle uw,v
\rangle$ for all $u,w \in N(v)$ pairs where $u < w$. In addition, a $\langle
uv,0 \rangle$ pair is returned where $u < v$ and $0$ signify that $u,v$ are
endpoints of an edge. The goal is to identify in the next round if a pair of
neighbors from $v$ are also adjacent and therefore form a triangle with $v$. The
emitted keys are sorted ordered pairs so that $(u,v)$ and $(v,u)$ will be
collected together in the next round. Overall, this round communicates
$O(m\sqrt{m})$ bits.

The map function in the second round is the identity function. The reduce
function returns triangle neighbor pairs from each triangle obtained in $\langle
uw, \{0,\{v\}\}\rangle$ where the value set contains a $0$, which denotes that
$u,w$ are triangle endpoints. The reduce function returns all possible pairs for
a $\{u,w,v\}$ triangle as key-value pairs but annotated with $1$ so these can be
distinguished as triangle neighbors in the next round. Thus, for each $v$ in the
value set, the reduce function returns the pairs,

\begin{displaymath}
  \langle v,(u,1) \rangle, \langle u,(v,1) \rangle, \langle v,(w,1) \rangle,
  \langle w,(v,1) \rangle, \langle u,(w,1) \rangle, \langle w,(u,1) \rangle. 
\end{displaymath}

The third round reads the original edge input in addition to the output from the
second round to complete the neighborhood set of each vertex. The rounds up to
this point discarded the adjacency information and kept only triangle
neighbors. The map function therefore maps $\langle (v,d(v)),(u,d(u))\rangle$ to
$\langle v,(u,0) \rangle$. Any $\langle v,(u,1) \rangle$ pair from the second
round is mapped to itself (identity). The reduce function gets the local
triangle count $\triangle(v)$ for the key $v$ by counting the number of $(u,1)$
values. It should be noted that there can be a multiplicity of a specific
$(u,1)$ value because $(u,v)$ can be in many triangles. Moreover, the second
round returned $v,u$ and $v,w$ from the same $\{v,u,w\}$ triangle and so that
triangle will be double counted in this round. Hence, the count of all $(u,1)$
values equates to $2\triangle(v)$ and must be halved for the output. At this
point, the triangle count for each $v$ is known, but in order to compute the
triangle centrality the triangle counts for each neighbor $u$ of $v$ is needed
and moreover these counts must be separated between triangle and non-triangle
neighbors. Therefore, the reduce function will return each neighbor $u$,
annotated if it is a triangle neighbor or not, with the triangle count of $v$ so
in the next round each vertex will be able to sum the triangle counts of its
neighbors accordingly. Recall that only triangle neighbors from the previous
round will have the number $1$ annotation, and edges from the original graph
have number $0$. These numbers are used to distinguish if a neighbor $u$ of $v$
is a triangle neighbor or not. Thus, if $u$ appears in the values with only $0$,
then it is not a triangle neighbor. The reduce function identifies the triangle
and non-triangle neighbors and then for each unique $u\in N(v)$ from the values
it returns $\langle u,(\triangle(v),\$) \rangle$ where

\begin{align*}
  \$ =
  \begin{cases}
    1 & \text{if u is a triangle neighbor,} \\
    0 & \text{otherwise.}
  \end{cases}
\end{align*}

Also, the key $v$ is returned with its triangle count since it is needed to
complete the triangle core sum on the closed triangle neighborhood,
$N_\triangle[v]$.

The fourth and final round calculates the triangle centrality for every
vertex. The map function is the identity. The reduce function takes $\langle
v,\{(\triangle(u),\$)\}\rangle$ and sums all $\triangle(u)$ in the values
separately for triangle neighbors ($\$=1$) and non-triangle neighbors
($\$=0$). Then, the triangle centrality is calculated for each $v\in V$ and
returned. We remark that the $\triangle(G)$ can be accumulated at the end of
third round and provided to each map and reduction function in the final round,
but we leave out the details for brevity.

\begin{algorithm}[H]
\caption{\label{alg:tricent_mapreduce}}
\begin{algorithmic}
  \Statex \textbf{Round 1}
  \Comment{return edge and neighbor pairs}
  \State\hspace{\algorithmicindent} \Call{Map}{}:
  \begin{displaymath}
    \bigl\langle \bigl( v,d(v) \bigr), \bigl( u,d(u) \bigr) \bigr\rangle
    \longrightarrow \langle v,u \rangle \tag{if $\pi(v) < \pi(u)$}
  \end{displaymath}
  \State\hspace{\algorithmicindent} \Call{Reduce}{}:
  \begin{align*}
    \langle v,N(v) \rangle &
    \longrightarrow 
    \begin{cases}
      \langle uv,0 \rangle & u < v \\
      \{\langle uw,v \rangle\} & u,w \in N(v), u < w
    \end{cases}
  \end{align*}

  \Statex \textbf{Round 2}
  \Comment{return oriented triangle edges}
  \State\hspace{\algorithmicindent} $\textproc{Map} \mapsto \text{Identity}$
  \State\hspace{\algorithmicindent} \Call{Reduce}{}:
  \begin{displaymath}
    \langle uw, \{0,\{v\}\} \rangle \longrightarrow
    \{\langle v,(u,1) \rangle,
    \langle u,(v,1) \rangle,
    \langle v,(w,1) \rangle,
    \langle w,(v,1) \rangle,
    \langle u,(w,1) \rangle,
    \langle w,(u,1) \rangle\}_{\forall v}
  \end{displaymath}

  \Statex \textbf{Round 3}
  \Comment{add $E$ again and compute triangle counts}
  \State\hspace{\algorithmicindent} \Call{Map}{}:
  \begin{align*}
    \langle v,(u,1) \rangle &\longrightarrow \langle v,(u,1) \rangle \\
    \bigl\langle \bigl( v,d(v) \bigr), \bigl( u,d(u) \bigr) \bigr\rangle
    &\longrightarrow \langle v,(u,0) \rangle
  \end{align*}
  \State\hspace{\algorithmicindent} \Call{Reduce}{}:
  \begin{align*}
    \bigl\langle v,\bigl(\{(u,0)|u\in N(v)\},(u,1),\ldots\bigr)\bigr\rangle
    \longrightarrow
    \begin{cases}
      \{\langle u,(\triangle(v),\$) \rangle\}_{\forall u\in N(v)} \\
      \langle v,(\triangle(v),1)\rangle
    \end{cases} \\
    \tag{$\triangle(v) = \frac{1}{2}\sum_u (u,1)$,
      \$ = 1 if u is a triangle neighbor, 0 otherwise}
  \end{align*}

  \Statex \textbf{Round 4}
  \Comment{calculate triangle centrality}
  \State\hspace{\algorithmicindent} $\textproc{Map} \mapsto \text{Identity}$
  \State\hspace{\algorithmicindent} \Call{Reduce}{}:
  \begin{displaymath}
    \langle v,\{(\triangle(u),\$)\} \rangle \longrightarrow
    \Bigl\{\Bigl\langle v,\frac{\frac{1}{3}x+y}{\triangle(G)}\Bigr\rangle\Bigr\}
    \tag{$x = \sum_u (\triangle(u),1), y = \sum_u (\triangle(u),0)$}
  \end{displaymath}
\end{algorithmic}
\end{algorithm}

\begin{theorem}
  \label{thm:tricent_mapreduce}
  Triangle centrality can be computed using MapReduce in four rounds and
  $O(m\sqrt{m})$ communication bits for all vertices in $G$.
\end{theorem}

\begin{proof}
  We will show that Algorithm~\ref{alg:tricent_mapreduce} achieves this
  claim. The accounting for the number of communication bits is as follows.

  The first round makes unique pairwise combinations of neighbor vertices for
  each vertex. Only degree-ordered edges $(v,u)$, where $\pi(v) < \pi(u)$, are
  used to create neighbor pairs, leaving half the edges. Then, each vertex $v$
  has $O(\sqrt{m})$ neighbors, all with higher degree; otherwise, it would lead
  to a contradiction of $\sum_{u\in N(v)} d(u) > O(m)$. This leads to
  ${\sum_{v\in V}\binom{d(v)}{2} \le \sqrt{m}\sum_{v\in V} d(v) \le
    O(m\sqrt{m})}$ unique neighbor pairs. Each degree-ordered edge is also
  returned with the number $0$, where there are $m$ instead of $2m$ edges. Thus,
  in total, there are $m+m\sqrt{m} = O(m\sqrt{m})$ bits communicated by this
  round.

  The second round takes all $O(m\sqrt{m})$ key-value pairs from the first round
  but ignores those key-value pairs that do not have the number $0$ in the
  values. This leaves only key-value pairs that correspond to triangle
  edges. The reduce step then returns the three edges of each triangle as
  directed pairs. This leads to a total of six triangle neighbor pairs for every
  triangle. Since there are $O(m\sqrt{m})$ triangles, then overall this round
  communicates $O(m\sqrt{m})$ bits.

  The third round combines all $2m$ edges with the triangle neighbor pairs from
  the second round to complete the neighborhood of each vertex. The triangle
  counts for triangle and non-triangle neighbors are computed in the reduce
  step. Then, for each vertex $v$, the unique neighbors $u\in N(v)$ are returned
  with $\triangle(v)$ and, respectively, the number $1$ or $0$ if $u$ is a
  triangle neighbor or not. Also $v$ is returned with $(\triangle(v),1)$ to
  complete the triangle core $N_\triangle[v]$ of $v$. Altogether this amounts to
  $2m + \sum_{v\in V} d(v)+1 \le n+4m \le O(n+m)$ bits communicated in this
  round.

  The fourth and final round totals the triangle counts from triangle neighbors
  and non-triangle neighbors separately and then computes the triangle
  centrality. This round returns each vertex with its triangle centrality and
  therefore communicates $O(n)$ bits.

  Altogether each round communicates $O(m\sqrt{m})$ bits and there are four
  rounds. Therefore, the algorithm takes four MapReduce rounds and communicates
  $O(m\sqrt{m})$ bits as claimed.
\end{proof}

\section{\label{sec:performance}Performance}
Next, we give performance results for computing triangle centrality on larger
graphs from the Stanford Network Analysis Project (SNAP)~\cite{bib:snapnets}.
Because the asymptotic runtime bound for triangle centrality differs from the
other centrality measures we have discussed, it would be meaningless to compare
empirical wallclock times. Therefore, we give benchmarks only for triangle
centrality.

We implement our main algorithm given in Algorithm~\ref{alg:tricent_main} in C++
and POSIX threads using the \textproc{triangleneighbor} procedure in
Appendix~\ref{sec:trineighbor_alt}. The benchmarks were run on a single
workstation with 256 GB of RAM and 28 Intel Xeon E5-2680
cores. Table~\ref{tbl:runtime} tabulates the runtime results.

\begin{table}[H]
\caption{\label{tbl:runtime} Runtime}
\centering
\sisetup{input-ignore={,},group-separator={,}}
\begin{tabular}{l
    S[table-format=8]
    *2{S[table-format=10]}
    S[table-format=2.3,round-mode=places]}
  \toprule
  \multicolumn{1}{c}{Graph}
  & {$n$ (vertices)}
  & {$m$ (edges)}
  & {$\triangle(G)$ (triangles)}
  & {wallclock (seconds)} \\
  \midrule
  com-Youtube & 1,134,890 & 2,987,624 & 3,056,386 & 0.231 \\
  as-Skitter & 1,696,415 & 11,095,298 & 28,769,868 & 0.589 \\
  com-LiveJournal & 3,997,962 & 34,681,189 & 177,820,130 & 1.51 \\
  com-Orkut & 3,072,441 & 117,185,083 & 627,584,181 & 4.87 \\
  com-Friendster & 65,608,366 & 1,806,067,135 & 4,173,724,142 & 68.4 \\
  \bottomrule
\end{tabular}
\end{table}

\section*{Acknowledgments}
The author is grateful to David G. Harris for helpful comments. The author also
thanks the anonymous reviewers for their suggestions that helped improve the
quality of the article.

\bibliographystyle{abbrvurl}
\bibliography{triangle_centrality}

\appendix
\section{\label{sec:review}Centrality Review}
The following is a brief review of popular centrality measures, specifically
closeness~\cite{bib:bavelas1950}, degree~\cite{bib:shaw1954,bib:nieminen1974},
eigenvector~\cite{bib:bonacich1972}, betweenness~\cite{bib:freeman1977}, and
PageRank~\cite{bib:pagerank1998}.

Closeness centrality was defined by Bavelas in 1950~\cite{bib:bavelas1950},
making one of the earliest measures for identifying central actors in a
network. In closeness centrality, a central vertex is closest to all others. The
closeness centrality of a vertex $v$ is the inverse of the average distance from
$v$ to all other vertices. Hence it is defined by,

\begin{displaymath}
  \cc(v) = \frac{n-1}{\sum_{u\in V} d_G(v,u)}.
\end{displaymath}

Some texts refer to the reciprocal of this equation for closeness
centrality~\cite{bib:freeman1979} where larger values indicate farness as
opposed to closeness. In matrix notation, the closeness centrality for all
vertices is,

\begin{displaymath}
  \cntvec{\cc} = (n-1)\frac{1}{\mat{A}^n\vec{1}}.
\end{displaymath}

The time complexity of closeness centrality is bounded by computing all-pairs
shortest paths and hence takes $O(mn)$ time.

Degree centrality was proposed in 1954 by Shaw~\cite{bib:shaw1954} and refined
later in 1974 by Nieminen~\cite{bib:nieminen1974}. It measures importance based
on the degree of a vertex, and therefore, important vertices are those with high
degree. The degree centrality of a vertex $v$ is then $\dc(v)=d(v)$. Using the
adjacency matrix $\mat{A}$, the degree centrality $\dc(v)$ is given by,

\begin{displaymath}
  \dc(v)
  = \sum_{u\in V} a(v,u)
  = \norm{\vec{A_v}}_1
  = \vec{A_v}^\transpose \vec{1}.
\end{displaymath}

\noindent
It then follows that all $\dc(v)$ values is computed by,

\begin{displaymath}
  \cntvec{\dc} = \mat{A}\vec{1}.
\end{displaymath}

This measure is very simple because it depends only on the degrees and thus is
also easy to compute, taking only $O(m)$ time.

Eigenvector centrality was formalized by Bonacich in
1972~\cite{bib:bonacich1972}. It measures the number and quality of connections
associated with a vertex based on the largest eigenvalue of the adjacency
matrix. Important vertices under this measure are those with many connections,
especially if the connections are to other important vertices. Eloquently put by
Borgatti and Everett, a central actor is therefore one who ``knows everybody who
is anybody''~\cite{bib:borgatti_everett2006}. The eigenvector centrality of a
vertex $v$ is given by the $v$th component of the eigenvector corresponding to
the largest eigenvalue in $\mat{A}$. Specifically, the relative eigenvector
centrality $\ev(v)$ for a vertex $v$ where $\lambda$ is the largest eigenvalue
of $\mat{A}$ is then,

\begin{displaymath}
  \ev(v) = \frac{1}{\lambda} \sum_{u\in V} a(v,u)\ev(u).
\end{displaymath}

Given a vector $\cntvec{\ev}$ to hold all $\ev(v)$, then it follows from the
familiar $\mat{A}\vec{x} = \lambda\vec{x}$ form of eigenvalue equations that the
eigenvector centrality in matrix notation for all vertices is,

\begin{displaymath}
  \cntvec{\ev} = \frac{1}{\lambda}\mat{A}\cntvec{\ev}.
\end{displaymath}

The advantage of this measure is it gives more weight to connections from
important vertices than those that are not important. Eigenvector centrality can
be computed in $O(n^3)$ time using the power iteration
method~\cite{bib:golub_vanloan2013}.

Betweenness centrality, introduced in 1977 by Freeman~\cite{bib:freeman1977},
measures importance by the number of shortest paths a vertex intersects. An
important vertex in this measure has a large fraction of shortest paths that
intersect it over all possible shortest paths. This implies that an important or
central vertex in a graph under this measure is one whose removal will disrupt
the flow of information to many other vertices. Let $\sigma_{ij}$ be the number
of shortest paths from $i$ to $j$, and $\sigma_{ij}(v)$ be the number of these
shortest paths that intersect $v$. Then, the betweenness centrality $\bc(v)$ for
a vertex $v$ is,

\begin{displaymath}
  \bc(v) = \sum_{i\ne v\ne j} \frac{\sigma_{ij}(v)}{\sigma_{ij}}.
\end{displaymath}

Computing betweenness centrality requires finding all-pairs shortest paths,
taking $O(mn)$ time~\cite{bib:brandes2001}.

PageRank centrality was published in 1998 by Brin and Page as the underlying
technology for the Google search engine~\cite{bib:pagerank1998}. The PageRank
centrality is a variant of eigenvector centrality. Importance is due to the
quantity and quality of links, just as in eigenvector centrality. But PageRank
also allows some randomness through a damping factor. Applied to web search,
PageRank centrality treats the World Wide Web as a graph where Web pages are
vertices and hyperlinks are edges. An important Web page in PageRank has many
hyperlinks and especially so if the hyperlinks are from other important Web
pages. The PageRank centrality $\pr(v)$ of a vertex $v$ and damping factor $d$
is,

\begin{displaymath}
  \pr(v) = \frac{1-d}{n} + d\sum_{u\in N(u)} \frac{\pr(u)}{d(u)}.
\end{displaymath}

Let the vector $\cntvec{\pr}$ hold all $\pr(v)$ values, then the algebraic form
is,

\begin{displaymath}
  \cntvec{\pr}=d\mat{M}\cntvec{\pr} + \frac{1-d}{n}\vec{1},
\end{displaymath}

\noindent where $\mat{M}=(\mat{K}^{-1}\mat{A})^\transpose$ and $\mat{K}$ is a
diagonal matrix in which the entries are the vertex degrees.

Computing the PageRank takes the same time as eigenvector centrality.

\subsection{\label{sec:trineighbor_alt}Alternative Triangle Neighbor Algorithm}
In Section~\ref{sec:algorithm}, we introduced our \textproc{triangleneighbor}
algorithm listed in Procedure~\ref{alg:trineighbor}. It is a practical and
deterministic algorithm for counting triangles and listing triangle neighbors in
optimal time and linear space. In this section, we describe an alternative
\textproc{triangleneighbor} algorithm that avoids linked lists to store triangle
neighbors of a vertex but takes $O(m\sqrt{m})$ time. Despite the slower runtime,
it may be practically more advantageous in parallel implementations.

\setcounter{algorithm}{1}
\begin{algorithm}[H]
  \floatname{algorithm}{Procedure}
  \caption{\label{alg:trineighbor_alt}Triangle Neighbor}
  \begin{algorithmic}[1]
    \Require $L_v$, array of size $|N_H(v)|$ for each $v\in V$,
    initialized to 0
    \Require sorted $N_H(v)$ for each $v\in V$
    \For{$v\in V$}
      \For{$u\in N_H(v)$}
        \State set $t := 0$
        \State set $i$ to the index of $u\in N_H(v)$
        \For{$w\in N_H(v)\cap N_H(u)$}
          \State set $t := 1$
          \State set $l$ to index of $w\in N_H(v)$ and $r$ to index of $w\in
          N_H(u)$
          \Comment{indices from set intersection}
          \State set $L_v(l) := 1$ and $L_u(r) := 1$
          \State increment
          $\triangle(v),\triangle(u),\triangle(w),\triangle(G)$
        \EndFor
        \If{$t \ne 0$}
          \State set $L_v(i) := 1$
        \EndIf
      \EndFor
    \EndFor
  \end{algorithmic}
\end{algorithm}

Let us call $v,u,w$, respectively, the \emph{low}, \emph{middle}, and
\emph{high} vertices. In Procedure~\ref{alg:trineighbor_alt}, we compute the set
intersection over the higher-ordered adjacency sets, $N_H(v)$. Then, any
$\{v,u,w\}$ triangle can only be reported by computing $N_H(v) \cap N_H(u)$
because $v$ is not in $N_H(u)$ or $N_H(w)$. Therefore, only the $\{v,u\}$ edge
can facilitate the $N_H(v) \cap N_H(u)$ set intersection. Meaning the set
intersections will only involve the higher-ordered adjacencies of the low and
middle vertices in each triangle, and only the high vertex can be returned by
the set intersections. Since all set intersections are over the higher-ordered
sets, it takes $O(m\sqrt{m})$ time. Then, the handling of marking triangle
neighbors to avoid duplication requires only simple arrays. We describe this in
detail next.

For every $v\in V$, an array $L_v$ of size $|N_H(v)|$ is used to mark the
position of each $u \in N_H(v)$ that is a triangle neighbor of $v$, as
determined by the set intersections. Thus, the higher-ordered triangle neighbors
of $v$ have the same position in $L_v$ as in $N_H(v)$. The set intersection
$N_H(v)\cap N_H(u)$ is computed as a linear scan over $N_H(v)$ and $N_H(u)$,
with two pointers to track the positions in $N_H(v),N_H(u)$. Then, any triangle
neighbor $w$ found by the set intersection must, respectively, correspond to
positions in $L_v$ and $L_u$. For example, suppose $w$ is found from $N_H(v)\cap
N_H(u)$ where $w$ is the $5$th neighbor in $N_H(v)$ and it is the $8$th neighbor
in $N_H(u)$. Then, $L_v(5)$ and $L_u(8)$ are set to 1 to mark that $w$ is a
triangle neighbor of $v$ and $u$. At the end of computing $N_H(v)\cap N_H(u)$,
if any common neighbor $w$ was found, then the position of $u\in N_H(v)$ is set
to 1 in $L_v$.

It is important to note that Procedure~\ref{alg:trineighbor_alt} does not output
the triangle neighbors explicitly, leaving it to subsequent
applications. Observe that only the low vertex $v$ in a $\{v,u,w\}$ triangle has
marked both $u,w$ as triangle neighbors, whereas $u$ has marked only $w$, and
$w$ has not marked either $v$ or $u$. But this is sufficient information to
identify all triangle neighbors of each vertex. For each $v$, simply scan over
all $u\in N_H(v)$ and if the corresponding position for $u$ in $L_v$ is 1, then
add $u$ to the triangle neighbor list of $v$ and also add $v$ to the triangle
neighbor list of $u$. This effectively ensures that the endpoints of each
triangle edge are paired and thus accomplishes the goal of optimally finding all
triangle neighbors of each vertex in linear space without hash tables, while
also computing all triangle counts. The correctness of
Procedure~\ref{alg:trineighbor_alt} is immediate because the endpoints of each
triangle edge are oppositely updated. The triangle counts follow naturally and
are essentially free in the computation.

\section{\label{sec:jaccard_similarity}Jaccard Similarity of Centralities}
The Jaccard index~\cite{bib:jaccard} is given by,

\begin{displaymath}
  J(i,j) = \frac{\lvert S_i \cap S_j \rvert}{\lvert S_i \cup S_j \rvert}
  = \frac{\lvert S_i \cap S_j \rvert}{\lvert S_i \rvert + \lvert S_j \rvert -
    \lvert S_i \cap S_j \rvert},
\end{displaymath}

\noindent
where $S_i,S_j$ are two sets. Hence, $J(i,j)=J(j,i)$ is a rational number in the
interval $[0,1]$ and denotes the fraction of overlap between the sets, where
values of 0 and 1, respectively, indicate disjoint and equivalent
sets. Therefore, we interpret two sets to be similar if the Jaccard index is
close to 1. The difference $1-J(i,j)$ suggests a distance between sets $S_i,S_j$
and can be used as a similarity metric.

Here $\lvert S_i \rvert = \lvert S_j \rvert$ for all pairs then the range of
$J(i,j)$ follows $\frac{c}{20-c}$, where $c=0..10$. Therefore, the possible
Jaccard index values are
$0,\frac{1}{19},\frac{2}{18},\frac{3}{17},\ldots,1$. These are unique so given
the real values it is not difficult to determine the rational form. Also note
since the sets are fixed size, then as the numerator increases the denominator
decreases by the same amount. Therefore, it is easy to see that a larger set
intersection corresponds to a larger Jaccard index.

For each centrality measure $i$, we denote $C_j$ as the closest centrality
measure $j$ to $i$. This is independent of the ordering within the top ten
rankings unless there are ties. In the case of a tie, we look at the first node
in $i$'s top ten list and choose the $j$ that ranks that same node higher than
the other tied measures, moving down $i$'s list on ties or misses. Empty entries
in Table~\ref{tbl:best_jaccard} indicate $J(i,j)=0$ for each pair of centrality
measures.

\begin{table}[H]
  \captionsetup{justification=centering}
  \caption{\label{tbl:best_jaccard} Most Similar Centrality Pairs for Graphs in
    Table~\ref{tbl:graphdata}\\
    (Jaccard Index $J(i,j)$ over Top $k=10$ Rankings)}
  \centering
  \footnotesize
  \begin{tabular}{S[table-format=2] *{12}{S[table-format=1.2]}}
    \toprule
    & \multicolumn{2}{c}{\tc}
    & \multicolumn{2}{c}{\bc}
    & \multicolumn{2}{c}{\cc}
    & \multicolumn{2}{c}{\dc}
    & \multicolumn{2}{c}{\ev}
    & \multicolumn{2}{c}{\pr} \\
    \cmidrule(r){2-3}
    \cmidrule(l){4-5}
    \cmidrule(l){6-7}
    \cmidrule(l){8-9}
    \cmidrule(l){10-11}
    \cmidrule(l){12-13}
    {Graph no.}
    & {$C_j$} & {$J(i,j)$}
    & {$C_j$} & {$J(i,j)$}
    & {$C_j$} & {$J(i,j)$}
    & {$C_j$} & {$J(i,j)$}
    & {$C_j$} & {$J(i,j)$}
    & {$C_j$} & {$J(i,j)$} \\
    \midrule
    1 & \ev & .82 & \cc & .67 & \bc & .67 & \ev & .67 & \tc & .82 & \dc & .67 \\
    2 & \cc & .67 & \cc & .82 & \bc & .82 & \ev & .82 & \dc & .82 & \dc & 1.00 \\
    3 & \ev & .54 & \cc & .43 & \bc & .43 & \pr & .67 & \dc & .54 & \dc & .67 \\
    4 & \ev & .54 & \pr & .67 & \dc & 1.00 & \cc & 1.00 & \tc & .54 & \bc & .67 \\
    5 & \ev & 1.00 & \pr & .67 & \dc & .54 & \pr & .54 & \tc & 1.00 & \bc & .67 \\
    6 & \dc & .67 & \cc & .54 & \bc & .54 & \pr & .82 & \tc & .54 & \dc & .82 \\
    7 & \ev & .54 & \pr & .82 & \dc & .82 & \cc & .82 & \dc & .67 & \bc & .82 \\
    8 & \dc & .25 & \cc & .33 & \bc & .33 & \pr & .67 & \dc & .33 & \dc & .67 \\
    9 & \ev & 1.00 & \dc & .54 & \dc & .54 & \pr & 1.00 & \tc & 1.00 & \dc & 1.00 \\
    10 & \ev & .82 & \pr & .82 & \dc & .82 & \pr & 1.00 & \tc & .82 & \dc & 1.00 \\
    11 & \dc & .25 & \cc & .54 & \bc & .54 & \pr & .33 & \dc & .18 & \bc & .33 \\
    12 & \ev & .82 & \cc & .25 & \bc & .25 & \pr & .54 & \tc & .82 & \dc & .54 \\
    13 & \ev & .67 & \pr & .43 & \bc & .33 & \pr & .54 & \tc & .67 & \dc & .54 \\
    14 & \ev & .67 & \pr & .33 & \dc & .54 & \tc & .67 & \tc & .67 & \cc & .54 \\
    15 & \ev & .82 & \pr & .67 & \tc & .67 & \pr & .67 & \tc & .82 & \dc & .67 \\
    16 & \ev & 1.00 & \pr & .43 & \dc & .43 & \cc & .43 & \tc & 1.00 & \bc & .43 \\
    17 & \ev & .33 & \pr & .33 & \bc & .18 & \ev & .82 & \dc & .82 & \bc & .33 \\
    18 & \ev & .82 & \pr & .43 & \bc & .25 & \bc & .25 & \tc & .82 & \bc & .43 \\
    19 & \dc & .33 & \dc & .25 & \bc & .18 & \pr & .54 & \tc & .11 & \dc & .54 \\
    20 & \ev & .05 & {} & {} & {} & {} & \pr & .05 & \tc & .05 & \dc & .05 \\
    \bottomrule
  \end{tabular}
\end{table}

\begin{table}[H]
  \caption{\label{tbl:allpairs_jaccard} All-Pairs Jaccard Similarity of Top
    $k=10$ Rankings over Graphs in Table~\ref{tbl:graphdata}}
  \centering
  \footnotesize
  \setlength{\arraycolsep}{3.33pt}
  \renewcommand{\frac}{\sfrac}
  \renewcommand{\arraystretch}{2}
  \[
  \begin{array}{*{21}{c}}
    \toprule
    \multicolumn{20}{c}{\text{Graph no.}} \\
    \cmidrule{2-21}
    J(i,j) & 1 & 2 & 3 & 4 & 5 & 6 & 7 & 8 & 9 & 10
    & 11 & 12 & 13 & 14 & 15 & 16 & 17 & 18 & 19 & 20 \\
    \midrule
    J(\tc,\bc) & \frac{5}{15} & \frac{8}{12} & \frac{1}{19} & \frac{5}{15}
    & \frac{4}{16} & \frac{5}{15} & \frac{5}{15} & 0
    & \frac{6}{14} & \frac{5}{15} & 0 & 0
    & 0 & \frac{3}{17} & \frac{2}{18} & \frac{2}{18} 
    & \frac{2}{18} & 0 & \frac{1}{19} & 0 \\
    J(\tc,\cc) & \frac{5}{15} & \frac{8}{12} & \frac{3}{17} & \frac{6}{14}
    & \frac{5}{15} & \frac{2}{18} & \frac{5}{15} & \frac{2}{18}
    & \frac{6}{14} & \frac{5}{15} & 0 & 0
    & 0 & \frac{7}{13} & \frac{8}{12} & \frac{5}{15}
    & \frac{2}{18} & 0 & \frac{2}{18} & 0 \\
    J(\tc,\dc) & \frac{7}{13} & \frac{7}{13} & \frac{5}{15} & \frac{6}{14}
    & \frac{7}{13} & \frac{8}{12} & \frac{6}{14} & \frac{4}{16}
    & \frac{9}{11} & \frac{5}{15} & \frac{4}{16} & \frac{1}{19}
    & 0 & \frac{8}{12} & \frac{4}{16} & \frac{4}{16}
    & \frac{5}{15} & \frac{1}{19} & \frac{5}{15} & \frac{1}{19} \\
    J(\tc,\ev) & \frac{9}{11} & \frac{8}{12} & \frac{7}{13} & \frac{7}{13}
    & 1 & \frac{7}{13} & \frac{7}{13} & \frac{3}{17}
    & 1 & \frac{9}{11} & 0 & \frac{9}{11}
    & \frac{8}{12} & \frac{8}{12} & \frac{9}{11} & 1
    & \frac{5}{15} & \frac{9}{11} & \frac{2}{18} & \frac{1}{19} \\
    J(\tc,\pr) & \frac{6}{14} & \frac{7}{13} & \frac{5}{15} & \frac{5}{15}
    & \frac{4}{16} & \frac{8}{12} & \frac{5}{15} & \frac{2}{18}
    & \frac{9}{11} & \frac{5}{15} & 0 & 0
    & 0 & \frac{5}{15} & \frac{2}{18} & \frac{2}{18}
    & \frac{2}{18} & 0 & \frac{3}{17} & 0 \\
    J(\bc,\cc) & \frac{8}{12} & \frac{9}{11} & \frac{6}{14} & \frac{8}{12}
    & \frac{6}{14} & \frac{7}{13} & \frac{7}{13} & \frac{5}{15}
    & \frac{6}{14} & \frac{9}{11} & \frac{7}{13} & \frac{4}{16}
    & \frac{5}{15} & \frac{4}{16} & \frac{3}{17} & \frac{5}{15}
    & \frac{3}{17} & \frac{4}{16} & \frac{3}{17} & 0 \\
    J(\bc,\dc) & \frac{8}{12} & \frac{8}{12} & \frac{6}{14} & \frac{8}{12}
    & \frac{7}{13} & \frac{5}{15} & \frac{7}{13} & \frac{1}{19}
    & \frac{7}{13} & \frac{9}{11} & \frac{2}{18} & 0
    & \frac{5}{15} & \frac{5}{15} & \frac{7}{13} & \frac{3}{17}
    & \frac{3}{17} & \frac{4}{16} & \frac{4}{16} & 0 \\
    J(\bc,\ev) & \frac{6}{14} & \frac{8}{12} & \frac{3}{17} & \frac{5}{15}
    & \frac{4}{16} & \frac{3}{17} & \frac{5}{15} & 0
    & \frac{6}{14} & \frac{5}{15} & \frac{1}{19} & 0
    & 0 & \frac{1}{19} & \frac{3}{17} & \frac{2}{18}
    & \frac{2}{18} & 0 & 0 & 0 \\
    J(\bc,\pr) & \frac{8}{12} & \frac{8}{12} & \frac{5}{15} & \frac{8}{12}
    & \frac{8}{12} & \frac{5}{15} & \frac{9}{11} & \frac{2}{18}
    & \frac{7}{13} & \frac{9}{11} & \frac{5}{15} & 0
    & \frac{6}{14} & \frac{5}{15} & \frac{8}{12} & \frac{6}{14}
    & \frac{5}{15} & \frac{6}{14} & \frac{3}{17} & 0 \\
    J(\cc,\dc) & \frac{7}{13} & \frac{9}{11} & \frac{5}{15} & 1
    & \frac{7}{13} & \frac{2}{18} & \frac{9}{11} & \frac{3}{17}
    & \frac{7}{13} & \frac{9}{11} & \frac{2}{18} & 0
    & \frac{4}{16} & \frac{7}{13} & \frac{4}{16} & \frac{6}{14}
    & \frac{2}{18} & \frac{2}{18} & \frac{3}{17} & 0 \\
    J(\cc,\ev) & \frac{6}{14} & \frac{9}{11} & \frac{4}{16} & \frac{7}{13}
    & \frac{5}{15} & \frac{3}{17} & \frac{7}{13} & \frac{2}{18}
    & \frac{6}{14} & \frac{5}{15} & \frac{1}{19} & 0
    & 0 & \frac{6}{14} & \frac{8}{12} & \frac{5}{15}
    & \frac{2}{18} & 0 & 0 & 0 \\
    J(\cc,\pr) & \frac{8}{12} & \frac{9}{11} & \frac{4}{16} & \frac{8}{12}
    & \frac{7}{13} & \frac{2}{18} & \frac{8}{12} & \frac{4}{16}
    & \frac{7}{13} & \frac{9}{11} & \frac{4}{16} & 0
    & \frac{3}{17} & \frac{7}{13} & \frac{2}{18} & \frac{3}{17}
    & \frac{2}{18} & \frac{1}{19} & \frac{2}{18} & 0 \\
    J(\dc,\ev) & \frac{8}{12} & \frac{9}{11} & \frac{7}{13} & \frac{7}{13}
    & \frac{7}{13} & \frac{5}{15} & \frac{8}{12} & \frac{5}{15}
    & \frac{9}{11} & \frac{5}{15} & \frac{3}{17} & \frac{1}{19}
    & 0 & \frac{6}{14} & \frac{5}{15} & \frac{4}{16}
    & \frac{9}{11} & \frac{1}{19} & \frac{1}{19} & \frac{1}{19} \\
    J(\dc,\pr) & \frac{8}{12} & 1 & \frac{8}{12} & \frac{8}{12}
    & \frac{7}{13} & \frac{9}{11} & \frac{8}{12} & \frac{8}{12}
    & 1 & 1 & \frac{5}{15} & \frac{7}{13}
    & \frac{7}{13} & \frac{7}{13} & \frac{8}{12} & \frac{2}{18}
    & \frac{3}{17} & \frac{4}{16} & \frac{7}{13} & \frac{1}{19} \\
    J(\ev,\pr) & \frac{7}{13} & \frac{9}{11} & \frac{6}{14} & \frac{5}{15}
    & \frac{4}{16} & \frac{5}{15} & \frac{6}{14} & \frac{3}{17}
    & \frac{9}{11} & \frac{5}{15} & \frac{3}{17} & 0
    & 0 & \frac{4}{16} & \frac{3}{17} & \frac{2}{18}
    & \frac{2}{18} & 0 & 0 & 0 \\
    \bottomrule
  \end{array}
  \]
  \end{table}

\end{document}